\documentclass[journal]{IEEEtran}
\usepackage[utf8]{inputenc}
\usepackage{graphicx}
\usepackage{color}
\usepackage{hyperref}
\usepackage{placeins}
\usepackage{float}
\usepackage{tabularx,colortbl}
\usepackage{amssymb}
\usepackage{amsthm}
\usepackage{cite}
\usepackage{amsmath}
\usepackage{euscript} 
\usepackage{cleveref}
\usepackage{multicol}       
\usepackage{multirow}       
\usepackage{array}          
\usepackage{makecell}
\usepackage{colortbl}
\usepackage{pifont}       
\usepackage{bbding}       
\usepackage{fontawesome}  
\usepackage{cases,subeqnarray}
\usepackage{bm,multirow,bigstrut}
\usepackage{textcomp}
\usepackage{latexsym,bm}
\usepackage{booktabs}
\usepackage{xcolor}
\usepackage{mathtools}
\usepackage{dsfont}
\usepackage{extarrows}
\usepackage{amsthm}
\usepackage[caption=false,font=footnotesize]{subfig}
\usepackage{makecell}
\definecolor{lightblue}{rgb}{0.93,0.95,1.0}
\usepackage[ruled,vlined,linesnumbered,noend]{algorithm2e}
\usepackage{mathrsfs}

\usepackage{booktabs} 
\usepackage{booktabs} 
\usepackage{longtable}
\usepackage{threeparttable}

\theoremstyle{plain}

\theoremstyle{plain}

\theoremstyle{definition}

\newtheorem{proposition}{\bf Proposition}
\newtheorem{definition}{Definition}
\newtheorem{assumption}{Assumption}

\newtheorem{theorem}{Theorem}
\newtheorem{lemma}{Lemma}

\newcommand{\E}{\mathbb{E}}
\newcommand{\btheta}{\boldsymbol{\theta}}
\newcommand{\bs}{\mathbf{s}}

\newcommand{\rin}{\bar{r}_{\text{in}}}
\newcommand{\rev}[1]{\textcolor{black}{#1}}

\DeclareRobustCommand{\scopeedit}[1]{#1}
\providecommand{\revise}[1]{\rev{#1}}
\providecommand{\rcom}[1]{\revise{#1}}
\newenvironment{revblock}{\begingroup\color{black}}{\endgroup}

\IEEEoverridecommandlockouts

\begin{document}
\title{Multi-Agent Reinforcement Learning for Movable Antenna-aided Cell-Free Massive MIMO Systems}
\author{Bokai Xu, Jiayi Zhang,~\IEEEmembership{Senior Member,~IEEE}, Shuaifei Chen,~\IEEEmembership{Member,~IEEE}, Ziheng Liu, Huahua Xiao,\\ Derrick Wing Kwan Ng,~\IEEEmembership{Fellow,~IEEE}, and Bo Ai,~\IEEEmembership{Fellow,~IEEE}\vspace{-1. em}
%
%
\thanks{B. Xu, J. Zhang, Z. Liu, and B. Ai are with the State Key Laboratory of Advanced Rail Autonomous Operation, and also with School of Electronic and Information Engineering, Beijing Jiaotong University, Beijing 100044, P. R. China; S. Chen is with the School of Communications and Information Engineering, Xi'an University of Posts and Telecommunications, Xi’an 710121, China. He is also with Purple Mountain Laboratories, Nanjing, 211111 China; Huahua Xiao is with ZTE Corporation, State Key Laboratory of Mobile Network and Mobile Multimedia Technology; D. W. K. Ng is with University of New South Wales, Sydney, NSW 2052, Australia. (Corresponding authors: Jiayi Zhang; Shuaifei Chen.)}}
\IEEEaftertitletext{\vspace{-2\baselineskip}}
\maketitle

\begin{abstract}
The inherent non-convex minimum-separation constraints introduced by movable antennas present a formidable challenge to the joint optimization of antenna positions and transmission strategies, rendering conventional methods computationally infeasible, particularly in large-scale cell-free massive multiple-input multiple-output (MIMO). In this work, we propose the graph-based learning individual intrinsic reward heterogeneous-agent proximal policy optimization (GLIIR-HAPPO) algorithm, a novel heterogeneous multi-agent reinforcement learning (MARL) framework that fundamentally overcomes this impasse by systematically decomposing the original coupled optimization into coordinated subproblems. 
To ensure tractability, we embed the non-convex geometric constraints into a penalty-augmented reward structure and develop a specialized geometric solver that enables the positioning agents to efficiently navigate the high-dimensional action space. 
Specifically, we propose an architecture featuring a dynamic-interaction graph critic for adaptive cross-role coordination, together with role-conditioned federated distillation that synchronizes same-role policies through compact actor-output statistics.
Beyond architectural design, we establish a rigorous theoretical analysis that derives monotonic performance improvement bounds and establishes convergence guarantees for the proposed bi-level optimization.
Numerical simulations demonstrate that our framework yields significant sum-rate improvements over state-of-the-art MARL schemes. Moreover, the performance of our advanced architecture closely approaches its fully centralized counterpart, while drastically reducing communication overhead.
\end{abstract}
\vspace{-0.2cm}

\IEEEpeerreviewmaketitle
\vspace{-0.45cm}
\section{Introduction}\label{sec:introduction}
Recently, the proliferation of wireless applications has led to an unprecedented surge in demand for higher throughput, improved spectral efficiency, and enhanced network capacity for future wireless communication systems beyond fifth-generation (5G) \cite{10155550,10578301, wang2026sparse}. Various emerging bandwidth-intensive applications such as high-quality video streaming and immersive virtual/augmented reality (VR/AR) have significantly contributed to this escalating demand. Recognizing this trend, researchers and engineers are directing their efforts towards the development of sixth-generation (6G) wireless communication systems and beyond \cite{10884927, xu2023resource}. In particular, multiple-input multiple-output (MIMO) communications have emerged as a foundational enabling technology for achieving these performance targets. However, in traditional cellular MIMO systems, users located at the cell edge often face severely degraded communication quality due to severe path loss, strong inter-cell interference, and limited spatial diversity.

The concept of cell-free (CF) networks, which eliminates the limitations imposed by traditional cell boundaries by coordinating a large number of geographically distributed access points (APs), has gained significant traction as a highly promising solution for improving network-wide throughput and enhancing a more uniform user experience \cite{[162]}. However, while CF MIMO technology can significantly improve both wireless channel capacity and system spectral efficiency (SE), it often entails non-negligible implementation complexity and increased energy consumption due to the need for numerous RF chains, dense AP deployments, and substantial fronthaul signaling/processing overhead. Furthermore, the desirable channel-hardening phenomenon manifests only under specific propagation and deployment regimes, such as scenarios with a relatively small path-loss exponent. It is also worth noting that favorable propagation characteristics are exhibited only among users that are well separated in space \cite{8379438}, which can limit the achievable spatial multiplexing gains in practical networks with heterogeneous and potentially clustered user distributions.

To address the limitations in CF massive MIMO systems, researchers have put forth various promising technologies. One such approach is antenna selection (AS), which reduces energy consumption by activating only those antennas that contribute the most to system throughput.
However, AS is inherently limited to a finite, discrete set of pre-deployed antenna locations. To unlock the full potential of the spatial domain, the concept of flexible-position antenna systems has emerged. A prominent realization of this concept is the fluid antenna system (FAS), which enables the flexible adjustment of radiating elements. 
The fluid antenna multiple access (FAMA) framework, for instance, empowers user-equipment devices (UEs) to exploit naturally occurring interference nulls by dynamically varying the antenna position, thereby establishing additional spatial degrees of freedom. In practice, this capability can be realized via liquid-based reconfigurable radiators or by switching among densely packed antenna ports \cite{9264694, 10146286}.

Parallel to the advancements in FAS, movable antenna (MA) technology has recently attracted considerable attention as another practical and robust implementation for realizing flexible-position MIMO \cite{10286328, shao20246dma}. While FAS effectively leverages spatial diversity through reconfigurable materials or port switching, MA systems typically utilize localized electromechanical drivers to continuously translate solid radiating elements within a prescribed region. In this work, we adopt the MA architecture for CF APs due to several specific advantages that align with our system design: First, MA can effectively leverage the available spatial DoF to its fullest potential by dynamically adjusting antenna positions, thereby improving SE. Second, MA can enhance channel gains by exploiting only a reduced number of repositionable antennas, leading to lower hardware requirements and enhanced energy efficiency (EE). Finally, MA can be relocated to spatial locations where multi-user channels become more orthogonal, facilitating the establishment of favorable propagation conditions.

\begin{revblock}
To unlock the full potential of wireless systems, optimizing transmission schemes is of utmost importance. In particular, MA-aided MIMO systems pose a significant challenge due to their highly coupled and geometric-dependent constraints.
Such tightly intertwined spatial and signal-domain relationships further complicate the underlying optimization task. For instance, in \cite{10243545}, the authors proposed an efficient alternating-optimization algorithm that iteratively optimizes the transmit covariance matrix and the positions of each transmit/receive MA, yielding a locally optimal but computationally intensive solution. Compared to traditional fixed-position antenna (FPA) systems, MA systems achieve significant gains in MIMO channel capacity due to their spatial flexibility. Another study \cite{10354003} employed a gradient-descent method to jointly optimize the positions of MAs for all users. However, from an optimization standpoint, MA-aided systems introduce inherently non-convex minimum-distance (collision-avoidance) constraints and may also involve non-convex objective functions. Indeed, these factors present challenges in solving the underlying optimization problems, especially at scale.
Existing research addresses these challenges through various approaches, such as successive convex approximation (SCA) \cite{10243545} or compressed sensing (CS)-based techniques \cite{10458417}. Nevertheless, these methods typically rely on surrogate convexifications of the objective and/or the non-convex distance constraints, or they introduce additional non-convex constraints (e.g., $l_0$-type sparsity constraints), which can incur higher computational burden and may lead to performance degradation.

Additionally, existing research on optimizing MA designs has primarily focused on centralized optimization algorithms. These approaches require the sharing of CSI and other information among all BSs, typically through a centralized design at the central processing unit (CPU). However, this centralized approach incurs substantial fronthaul signaling and data-exchange overhead \cite{peng2024}.
In real-world deployments, BSs cannot engage in frequent message exchanges due to stringent fronthaul bandwidth/latency constraints and limited link capacity, making iterative centralized procedures less favorable. 
This raises a crucial question: How can we effectively optimize the transmission strategies of multiple BSs with minimal or even no message exchanges \cite{liu2024survey}?
In the realm of distributed optimization, multi-agent reinforcement learning (MARL) has emerged as a highly promising framework with rapidly growing adoption \cite{lowe2017multi}. In practice, MARL naturally complements the distributed nature of cell-free mMIMO systems by enabling multiple agents (e.g., mobile-BSs or UEs) to learn coordinated policies from local observations and optimize their strategies collaboratively. \scopeedit{Under centralized training with decentralized execution, agents can use aggregated information during training and local observations for online policy execution.} However, conventional MARL-based schemes face fundamental challenges, including poor sample efficiency, limited scalability, and unstable learning dynamics, due to non-stationarity and difficult credit assignment when agents are driven solely by a global team reward. Therefore, it is necessary to design a decoupled learning and coordination architecture to improve the robustness and performance of MARL \cite{10559632}. In addition to these algorithmic challenges, existing literature has yet to adequately resolve several critical implementation difficulties: 1) the computational intractability arising from non-convex, geometry-coupled constraints; 2) the prohibitive communication overhead required for effective coordination in large-scale networks; and 3) the amplified learning instability when heterogeneous agents attempt to collaborate under sparse global rewards.
\end{revblock}

\begin{figure*}[t]
\captionsetup{skip=-5pt}
	\centering
	\includegraphics[width=0.8\textwidth]{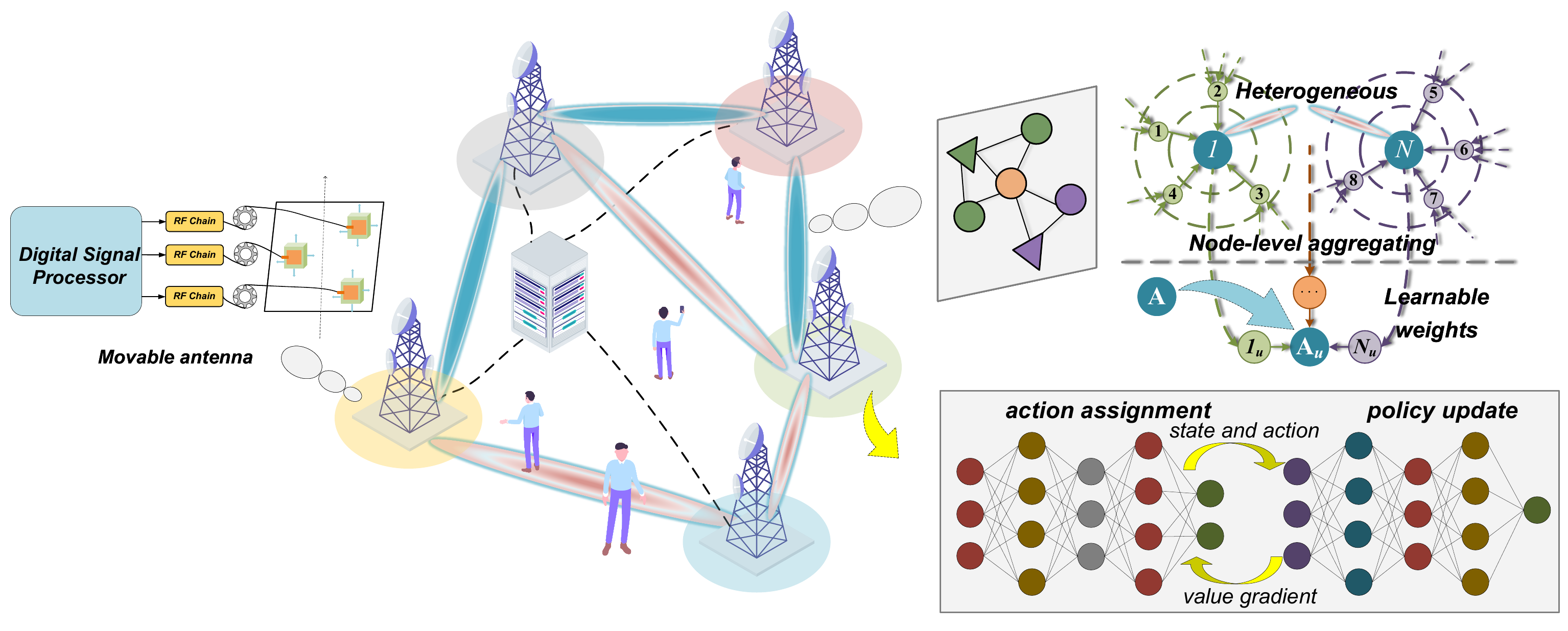}
	\caption{System model and heterogeneous multi-agent architecture for MA-aided cell-free massive MIMO.} 
\label{fig:sys_model}
\end{figure*}



To fill in the above gap, in this paper, we address the following fundamental question: \emph{How can we design \scopeedit{a heterogeneous multi-agent optimization framework} that simultaneously resolves the geometric constraints of antenna positioning and the interference management of power control?}
\par
\label{par:scope_clarification}
\scopeedit{This work investigates the joint optimization of MA positioning and power allocation under the adopted channel and geometric models. Problem $(\mathcal{P})$ maximizes the nominal sum rate subject to power and geometric constraints, whereas the MARL extrinsic team reward combines the nominal sum rate with a geometric feasibility penalty. Practical MA non-idealities, including movement delay, actuation energy, positioning uncertainty, finite movement resolution, and CSI/control overhead, are considered through simplified modeling and supplementary sensitivity analysis.}
Our contributions are summarized as follows.
\begin{itemize}
\item We introduce a novel heterogeneous multi-agent collaboration learning framework, powered by a graph-based learning individual intrinsic reward heterogeneous-agent proximal policy optimization (GLIIR-HAPPO) algorithm to overcome the intractability of monolithic centralized optimizers in MA-aided CF-MIMO networks. This paradigm fundamentally decomposes the joint optimization of antenna positions and transmission strategies into role-specialized subtasks, handled respectively by MA-positioning agents and power-control agents, \scopeedit{enabling role-specialized policy learning with centralized training and decentralized online execution.}

\item We tackle the intractable non-convex minimum-distance constraints between antennas through a dual strategy. First, a penalty-augmented mechanism is integrated into the reward function to drive the learning process toward constraint satisfaction. Second, a specialized geometric solver empowers positioning agents to efficiently navigate the high-dimensional action space and identify feasible, collision-free locations, transforming an otherwise intractable geometric problem into a manageable optimization task.

\item \scopeedit{To reduce learning-related coordination overhead, we combine a dynamic-interaction graph critic with role-conditioned federated distillation. The critic coordinates agents across roles, while distillation exchanges compact same-role actor-output statistics. We establish a trust-region monotonic-improvement bound and a bi-level intrinsic-reward convergence result under the stated assumptions. Simulations achieve near-centralized performance with reduced learning-related communication.}
\end{itemize}

The remainder of this paper is structured as follows. Section \uppercase\expandafter{\romannumeral2} describes the system model of the MA-aided CF massive MIMO network and formulates the joint optimization problem under non-convex minimum-distance constraints. In Section \uppercase\expandafter{\romannumeral3}, we present the proposed heterogeneous multi-agent collaboration framework, which integrates a penalty-guided geometric solver with the GLIIR-HAPPO algorithm to decouple and solve the resulting sub-tasks. Section \uppercase\expandafter{\romannumeral4} provides the theoretical analysis of the proposed framework, establishing monotonic performance bounds and convergence guarantees. Numerical simulation results that validate the effectiveness of the proposed approach are presented in Section \uppercase\expandafter{\romannumeral5}. Finally, Section \uppercase\expandafter{\romannumeral6} concludes the paper.

\emph{{Notations}}:
In this paper, scalars are denoted by non-bold italic letters $a$ or $A$, while vectors and matrices are represented by boldface lowercase $\boldsymbol{a}$ and boldface uppercase letters $\boldsymbol{A}$, respectively. The superscripts $(\cdot)^T$ and $(\cdot)^H$ denote the transpose and the conjugate transpose operations, respectively. $\mathbb{R}^{M \times N}$ and $\mathbb{C}^{M \times N}$ represent the spaces of $M \times N$ real-valued and complex-valued matrices, respectively. For a vector $\boldsymbol{x}$, $\|\boldsymbol{x}\|_2$ denotes its Euclidean norm (i.e., the $\ell_2$-norm). Calligraphic letters $\mathcal{A}$ represent sets and $|\mathcal{A}|$ denotes the cardinality of the set $\mathcal{A}$, or the absolute value if applied to a scalar. $\mathbb{E}[\cdot]$ denotes the statistical expectation operator, and $\nabla$ represents the gradient operator. 

\vspace{-0.15cm}
\section{System Model and Joint Optimization Problem}\label{sec:system_model}
We consider a CF MIMO network consisting of $L$ geographically distributed access points (APs), each equipped with $N$ MAs. The APs are connected via fronthaul connections to a CPU, as illustrated in Fig.~\ref{fig:sys_model}.
\begin{revblock}
Let $N$ denote the number of movable radiating elements (i.e., MAs) at each AP. The channel vector from AP $l$ to UE $k$ is denoted by $\boldsymbol{h}_{lk}(\boldsymbol{\overline{r}}_{l})\in\mathbb{C}^{N\times1}$, which contains the channel coefficients of the $N$ elements.
\end{revblock}
There are $K$ single-antenna UEs in the network and the channel between the $l$-th AP, $l\in \{1,\dots,L  \}$ and the $k$-th UE, $k\in \{1,\dots,K  \}$ is denoted by $\boldsymbol{h}_{lk}(\boldsymbol{\overline{r}}_{l}) \in \mathbb{C}^{N \times 1}$, the position of the $n$-th MA at the $l$-th AP by $\boldsymbol{r}_{ln}$.
We use $\boldsymbol{\overline{r}}_{l} \triangleq[\boldsymbol{r}_{l1}^{T}, \ldots, \boldsymbol{r}_{lN}^{T}]^{T} \in \mathbb{R}^{2N \times 1}$ to denote the aggregated position vector of the $l$-th AP.
The significance of MA stems from the spatial channel model. In this paper, we employ a multipath channel model, which enables us to accurately capture the position-dependent structure of MA channels \cite{10458417, 10354003}. In particular, the downlink channel vector $\boldsymbol{h}_{lk}(\boldsymbol{\overline{r}}_{l}) \in \mathbb{C}^{N \times 1}$ consists of $P$ paths from the $l$-th AP to the $k$-th UE and is given by
\vspace{-0.5 em}\begin{equation}
\boldsymbol{h}_{lk}(\boldsymbol{\overline{r}}_{l})=\sqrt{\frac{1}{P}} \sum_{p=1}^{P} \beta_{lk,p} \boldsymbol{a}\left(\theta_{lk,p}, \phi_{lk,p}, \boldsymbol{\overline{r}}_{l}\right),
\vspace{-0.25 em}\end{equation}
where $\beta_{lk,p}\in\mathbb{C}$ is the complex path gain of the $p$-th path \revise{for the link between the $l$-th AP and the $k$-th UE}, and $P$ is the number of spatial channel paths.
In addition, $\boldsymbol{a}\left(\theta_{lk,p}, \phi_{lk,p}, \boldsymbol{\overline{r}}_{l}\right)\in \mathbb{C}^{N \times 1}$ represents the transmit array
response vector. Variable $\phi_{lk,p} \triangleq \sin \left(\varphi_{lk,p}\right) \sin \left(\vartheta_{lk,p}\right)$, and $\theta_{lk,p} \triangleq \cos \left(\vartheta_{lk,p}\right)$ denote the virtual azimuth and elevation angles of the $p$-th path \revise{for the link between the $l$-th AP and the $k$-th UE}, respectively.
The transmit array response vector \cite{11058968, chen2026channel} is
\vspace{-0.5 em}\begin{equation}
\begin{aligned}
&\boldsymbol{a}\left(\theta_{lk, p}, \phi_{lk,p}, \boldsymbol{\overline{r}}_{l}\right)  =\left[e^{j \frac{2 \pi}{\lambda}\left(\phi_{lk,p} x_{l,1}+\theta_{lk,p} z_{l,1}\right)}, \cdots\right. \\
& \left.e^{j \frac{2 \pi}{\lambda}\left(\phi_{lk,p} x_{l,n_{x}}+\theta_{lk,p} z_{l,n_{z}}\right)}, \cdots, e^{j \frac{2 \pi}{\lambda}\left(\phi_{lk,p} x_{l,N_{x}}+\theta_{lk,p} z_{l,N_{z}}\right)}\right]^{T},
\end{aligned}
\end{equation}
where $\left(\boldsymbol{r}_{ln_{x}}, \boldsymbol{r}_{ln_{z}}\right)$ is the position of the $\left(n_{x}, n_{z}\right)$-th antenna in the $x$-$z$ plane at the $l$-th AP.
During the downlink data transmission, all antennas of all $L$ APs simultaneously serve all $K$ UEs via the same time-frequency resources, and the received downlink signal $y_k$ at the $k$-th UE is given by
\vspace{-0.5 em}\begin{equation}
y_{k}=\sum_{l=1}^{L} \sum_{i=1}^{K} \sqrt{\rho_{l i}} \boldsymbol{h}_{lk}^{H}(\boldsymbol{\overline{r}}_{l}) \boldsymbol{w}_{l i} s_{i}+n_{k},
\vspace{-0.25 em}\end{equation}
where $\rho_{l i} \geqslant 0$ denotes the downlink power allocated to the $i$-th UE by $l$-th AP, and $\boldsymbol{w}_{l i} \in \mathbb{C}^{N \times 1}$ represents the normalized precoding vector, satisfying $\left\|\boldsymbol{w}_{l i}\right\|^{2}=1$ and $\boldsymbol{w}_{l i}=\tilde{\boldsymbol{w}}_{l i} /\left\|\tilde{\boldsymbol{w}}_{l i}\right\|$
with the unnormalized precoding vector $\tilde{\boldsymbol{w}}_{l i}$, and $s_i$ denotes the symbol intended for the $i$-th UE with unit average power, i.e., $\mathbb{E}\{|s_i|^{2}\}=1$.
\begin{revblock}
In this work, the precoder is treated as a fixed deterministic module rather than an additional MARL action. For AP $l$, define $\boldsymbol{H}_{l}(\boldsymbol{\overline{r}}_{l})=[\boldsymbol{h}_{l1}(\boldsymbol{\overline{r}}_{l}),\ldots,\boldsymbol{h}_{lK}(\boldsymbol{\overline{r}}_{l})]\in\mathbb{C}^{N\times K}$. The unnormalized RZF precoder is computed as
\begin{equation}
\widetilde{\boldsymbol{W}}_{l}
=
\boldsymbol{H}_{l}(\boldsymbol{\overline{r}}_{l})
\left(
\boldsymbol{H}_{l}^{H}(\boldsymbol{\overline{r}}_{l})
\boldsymbol{H}_{l}(\boldsymbol{\overline{r}}_{l})
+\delta\boldsymbol{I}_{K}
\right)^{-1},
\end{equation}
where $\delta>0$ is the regularization factor, and $\boldsymbol{w}_{li}$ is obtained by normalizing the $i$-th column of $\widetilde{\boldsymbol{W}}_{l}$. This design keeps the MARL action space focused on power allocation and MA displacement, avoiding high-dimensional complex-valued beamforming actions. Since $\boldsymbol{H}_{l}(\boldsymbol{\overline{r}}_{l})$ depends on the MA positions, optimizing $\boldsymbol{\overline{r}}_{l}$ still changes the Gram matrix $\boldsymbol{H}_{l}^{H}(\boldsymbol{\overline{r}}_{l})\boldsymbol{H}_{l}(\boldsymbol{\overline{r}}_{l})$ and can indirectly improve channel orthogonality and RZF interference suppression.
\end{revblock}
$n_k \sim \mathcal{CN}(0, \sigma_k^2)$ denotes the additive white Gaussian noise (AWGN) at UE $k$, where $\sigma_k^2$ is the noise power. 
The SINR for the $k$-th UE can be expressed as
\vspace{-0.5 em}\begin{equation}
\operatorname{SINR}_{k}=\frac{\left|\sum_{l=1}^{L} \sqrt{\rho_{l k}} \boldsymbol{h}_{l k}^{H}(\boldsymbol{\overline{r}}_{l}) \boldsymbol{w}_{l k}\right|^{2}}{\sum_{i \neq k}^{K}\left|\sum_{l=1}^{L} \sqrt{\rho_{l i}} \boldsymbol{h}_{l k}^{H}(\boldsymbol{\overline{r}}_{l}) \boldsymbol{w}_{l i}\right|^{2}+\sigma_{k}^{2}}.
\end{equation}
The achievable rate of the $k$-th UE is given by
\vspace{-0.5 em}\begin{equation}
R_{k}=\log _{2}\left(1+\operatorname{SINR}_{k}\right).
\end{equation}
Our primary design objective is to maximize the overall network performance by jointly optimizing the downlink power allocation and the physical positions of the movable antennas. This design is formulated as a sum-rate maximization problem, which can be expressed as follows:
\vspace{-0.5 em}\begin{subequations}
\begin{align}
(\mathcal{P}):  \max _{\rho_{l k}, \boldsymbol{r}_{ln}}&\sum_{k=1}^{K} R_k\\
\text { s.t. } &\sum_{k=1}^{K} \rho_{l k} \leq P_{l}^{\max }, \quad \rho_{l k} \geq 0, \quad \forall l, k, \label{p_max}\\ 
&  \boldsymbol{r}_{ln} \in \mathcal{C}, \quad \forall l, n, \label{r_ln_C}\\
& \left\|\boldsymbol{r}_{ln}-\boldsymbol{r}_{lm}\right\|_{2} \geq D, \quad \forall l, n, m, (m \neq n),\label{nc_const}
\end{align}
\end{subequations}
where constraint~\eqref{p_max} ensures that the total transmit power from each AP does not exceed its maximum $P_{l}^{\max }$. In constraint~\eqref{r_ln_C}, the region available for the movement of MAs is represented by $\mathcal{C}$, while constant $D$ in constraint~\eqref{nc_const} specifies the minimum inter-antenna separation distance.

\begin{revblock}
\label{par:practical_model_scope}
Problem $(\mathcal{P})$ represents the nominal joint power-allocation and MA-positioning design. \scopeedit{The following simplified models support a separate sensitivity analysis and are not calibrated to a specific actuator or measured CSI-acquisition protocol.} In practical MA-aided systems, however, antenna repositioning is not instantaneous and the channel must be reacquired after position updates. To capture these effects, we assume that MA positions are updated over a slow window $T_{\rm upd}$ and denote the movement time by $\tau_{\rm mov,t}=\max_{l,n}\|\Delta\boldsymbol r_{ln,t}\|_2/v_{\max}$. The resulting effective sum rate is then written as
\vspace{-0.5 em}\begin{equation}
R_{\rm sum}^{\rm eff}(t)
=
\left[
1-\frac{\tau_{\rm p}+\tau_{\rm fb}+\tau_{\rm csi}+\tau_{\rm mov,t}}{T_{\rm upd}}
\right]^+
\sum_{k=1}^{K}R_{k,t}.
\label{eq:practical_eff_rate_ee}
\end{equation}
In \eqref{eq:practical_eff_rate_ee}, $\tau_{\rm p}$, $\tau_{\rm fb}$, and $\tau_{\rm csi}$ denote the pilot-training, feedback/control, and CSI-reacquisition durations, respectively. When EE is considered, the same effective-rate model can be normalized by the total consumed power as $\mathrm{EE}_t=R_{\rm sum}^{\rm eff}(t)/P_{\rm tot,t}$. \scopeedit{Effective rate and EE serve only as supplementary evaluation metrics and are not incorporated into the nominal sum-rate objective in $(\mathcal{P})$ or the proposed MARL reward.}

The total power consumption should include not only the transmit power but also circuit, RF-chain, actuation, and control power terms \cite{ngo2018total_ee_cf}. Accordingly, we model it as
\begin{equation}
\begin{aligned}
P_{\rm tot,t}
&=
\sum_{l,k}\rho_{lk,t}+LP_{\rm c}+LNP_{\rm RF}
\\
&\quad+\frac{1}{T_{\rm upd}}
\sum_{l,n}\!\left(c_1\|\Delta\boldsymbol r_{ln,t}\|_2
+c_2\|\Delta\boldsymbol r_{ln,t}\|_2^2\right)
+P_{\rm ctrl,t}.
\end{aligned}
\label{eq:practical_power_act}
\end{equation}
Here, the summation term represents the actuation energy amortized over the MA update window, while $P_{\rm ctrl,t}$ accounts for the control and feedback signaling power.

We further account for positioning uncertainty and finite movement resolution. Specifically, if the actual MA position is $\widetilde{\boldsymbol r}_{ln}=\boldsymbol r_{ln}+\boldsymbol e_{ln}$ with $\|\boldsymbol e_{ln}\|_2\le\epsilon_{\rm pos}$, then the robust separation margin is set as $D_{\rm eff}=D+2\epsilon_{\rm pos}$ and the constraint $\|\boldsymbol r_{ln}-\boldsymbol r_{lm}\|_2\ge D_{\rm eff}$ is imposed. After projection onto the feasible movement region, the finite-resolution MA update is given by
\begin{equation}
\boldsymbol r_{ln,t+1}
=
\mathcal Q_{\Delta_r}\!\left(
\Pi_{\mathcal C}
\left(\boldsymbol r_{ln,t}+\Delta_{\max}\boldsymbol a_{ln,t}^{r}\right)
\right).
\label{eq:practical_error_resolution}
\end{equation}
\scopeedit{In the sensitivity evaluation, finite resolution is approximated by uniformly distributed coordinate residuals in $[-\Delta_r/2,\Delta_r/2]$.}
\scopeedit{Finally, the overall overhead can be decomposed as $C_{\rm total}=C_{\rm pilot}+C_{\rm CSI}+C_{\rm control}+C_{\rm MARL}$. FD reduces only the learning-related term $C_{\rm MARL}$, whereas MA relocation still requires CSI reacquisition and control signaling \cite{xiao2024ma_channel_estimation}.}
\end{revblock}

The joint optimization in problem $(\mathcal{P})$ is highly challenging due to the strong coupling among variables and, most notably, the non-convex minimum distance constraint~\eqref{nc_const} and objective function, which renders conventional optimization methods ineffective. 
\scopeedit{We develop a heterogeneous multi-agent framework with centralized training and decentralized online execution. During training, the graph critic uses aggregated observations across APs, whereas each actor selects actions from its local observation during online execution. To handle the coupled geometric constraints, we reformulate the original design through variable splitting and quadratic-penalty augmentation.}


\section{Problem Reconstruction and Heterogeneous Multi-Agent Modeling}\label{Pro_HMARL}

\scopeedit{In this section, we present a heterogeneous multi-agent framework to handle problem $(\mathcal{P})$ with centralized training and decentralized online execution.} We first transform problem $(\mathcal{P})$ into a MARL-compatible form, after which we detail the agent structure, state representations, action design, and reward formulation.
\subsection{Problem Reconstruction}\label{pro_recons}
To decouple the non-convex constraint \eqref{nc_const} from the objective function $\quad\sum_{k=1}^{K} R_k$, we utilize variable splitting by introducing auxiliary variables $\left\{\boldsymbol{z}_{ln}=\boldsymbol{r}_{ln}\right\}_{\forall l, n}$. By substituting $\left\{\boldsymbol{z}_{ln}\right\}_{\forall l, n}$ for $\left\{\boldsymbol{r}_{ln}\right\}_{\forall l, n}$ in \eqref{nc_const}, and introducing a penalty term $\varrho  \sum_{l=1}^{L}\sum_{n=1}^{N}\left\|\boldsymbol{r}_{ln}-\boldsymbol{z}_{ln}\right\|_{2}^{2}$ into the objective function, the problem $(\mathcal{P})$ can be transformed into
\begin{subequations}
\begin{align}
(\mathcal{P} \text{1}):  \max _{\rho_{l k}, \boldsymbol{r}_{ln}, \boldsymbol{z}_{ln}} &\sum_{k=1}^{K} R_k-\varrho\sum_{l=1}^{L}\sum_{n=1}^{N}\!\left\|\boldsymbol{r}_{ln}-\boldsymbol{z}_{ln}\right\|_{2}^{2}\\
\text { s.t. } &\sum_{k=1}^{K} \rho_{l k} \leq P_{l}^{\max }, \quad \rho_{l k} \geq 0, \forall l, k,\\
&\boldsymbol{r}_{ln} \in \mathcal{C}, \quad \forall l, n,\\
& \left\|\boldsymbol{z}_{ln}-\boldsymbol{z}_{lm}\right\|_{2} \geq D,  \forall l, n, m, m \neq n,
\end{align}
\end{subequations}
\begin{revblock}
where $\varrho>0$ denotes the quadratic-penalty coefficient, distinct from the transmit power $\rho_{lk}$. In simulations, we set $\varrho_0=0.1$ and update it once per episode according to the average minimum-distance violation
\begin{equation}
\label{eq:penalty_violation}
\nu_e
=
\frac{1}{LN(N-1)}
\sum_{l=1}^{L}
\sum_{n=1}^{N}
\sum_{m\neq n}
\left[
D-
\|\boldsymbol{r}_{ln}-\boldsymbol{r}_{lm}\|_2
\right]^+,
\end{equation}
where $[x]^+=\max\{x,0\}$. The adaptive update rule is
\begin{equation}
\label{eq:penalty_update}
\varrho_{e+1}
=
\begin{cases}
\min\{\eta_{\varrho}\varrho_e,\varrho_{\max}\},
&
\nu_e>\epsilon_{\rm feas},\\
\varrho_e,
&
\nu_e\le \epsilon_{\rm feas},
\end{cases}
\end{equation}
with $\eta_{\varrho}=1.05$, $\varrho_{\max}=20$, and $\epsilon_{\rm feas}=10^{-3}D$. The penalty is therefore increased only when the violation exceeds the feasibility threshold, which mitigates early-stage penalty domination.
\end{revblock}
Fix an AP index $l$ and an antenna index $n$, the subproblem of $(\mathcal{P}\text{1})$ can be written as
\begin{equation}\label{eq:P1cm}
\begin{aligned}
(\mathcal{P}\text{1-a}):
& \min_{\bm z_{ln}\in\mathbb{R}^2}\ \bigl\|\bm r_{ln}-\bm z_{ln}\bigr\|_2^2 \\
\text{s.t.}\quad
& \bigl\|\bm z_{ln}-\bm z_{lm}\bigr\|_2 \ge D,\quad
  \forall\, l,m,n\ (m\neq n),
\end{aligned}
\end{equation}
where $\bm r_{ln}\in\mathbb{R}^2$ is the unconstrained update of the physical MA position at AP $l$ coming from the penalized master problem, while $\{\bm z_{lm}\}_{m\neq n}$ are the auxiliary variables used to enforce the minimum separation $D$.
Since subproblem (\ref{eq:P1cm}) involves only variables associated with AP index $l$, the AP-level subproblems are independent and can be solved in parallel. In the following, we characterize the solution of all AP-level subproblems by exploiting the geometric structure of~\eqref{eq:P1cm}.

\begin{definition}[Candidate points on active boundaries]
Define the index set of currently violated distance constraints
\begin{equation}
\mathcal{D}_{ln} \ \triangleq\ \Bigl\{m\in\mathcal{N}\setminus\{n\}\ :\ \|\bm r_{ln}-\bm z_{lm}\|_2<D\Bigr\}.
\end{equation}
For every $n\in\mathcal{N}$, let
\begin{equation}
\mathrm{Circle}_{ln} \ \triangleq\ \bigl\{\bm x\in\mathbb{R}^2: \|\bm x-\bm z_{ln}\|_2=D\bigr\}
\end{equation}
be the circle that bounds the feasible exterior region.
Let $\mathrm{Line}_{l,\ell n}$ denote the line through $\boldsymbol{z}_{l\ell}$ and $\boldsymbol{r}_{ln}$.
Define
\begin{equation}
\mathcal{W}_{l\ell}\! \triangleq\!
\left\{\boldsymbol{w}\in\mathbb R^2 \middle|
\begin{aligned}
& \boldsymbol{w} \in \mathrm{Circle}_{l\ell} \cap \mathrm{Circle}_{l\ell'},\exists\,\ell'\in \mathcal{N}\setminus\{n,\ell\},\\
& \|\boldsymbol{w}-\boldsymbol{z}_{l\ell'}\|_2 \ge D,\quad \forall\,\ell'\in \mathcal{N} \setminus\{n,\ell\}
\end{aligned}
\right\}
\end{equation}
and
\begin{equation}\label{U_l}
\mathcal{U}_{l\ell}\! \triangleq \!
\left\{\boldsymbol{u}\in\mathbb R^{2}\ \Bigm|\ 
\begin{aligned}
& \boldsymbol{u} \in \mathrm{Circle}_{l\ell}\cap \mathrm{Line}_{l,\ell n},\\
& \|\boldsymbol{u}-\boldsymbol{z}_{l\ell'}\|_{2}\ge D, \quad
  \forall\,\ell'\in \mathcal{N} \setminus\{n,\ell\}
\end{aligned}
\right\}.
\end{equation}
If the centers of $\mathrm{Circle}_{l\ell}$ and 
$\mathrm{Circle}_{l\ell'}$ coincide, the intersection 
$\mathrm{Circle}_{l\ell} \cap \mathrm{Circle}_{l\ell'}$ is regarded as $\varnothing$ in $\mathcal{W}_{l\ell}$; 
if $\boldsymbol{z}_{l\ell}=\boldsymbol{r}_{l\ell'}$, then $\mathcal{U}_{l\ell}$ is regarded as $\varnothing$.
\end{definition}

\begin{lemma}\label{lem:existence}
The feasible set $\mathcal{F}_{ln}=\bigcap_{m\neq n}\{\bm z:\ \|\bm z-\bm z_{lm}\|_2\ge D\}$ is closed and nonempty. The objective $f(\bm z)=\|\bm z-\bm r_{ln}\|_2^2$ is coercive. Hence (\ref{eq:P1cm}) admits at least one global minimizer $\bm z_{ln}^\star$. Moreover:
\begin{itemize}
\item[(i)] If $\mathcal{D}_{ln}=\varnothing$ (i.e., $\bm r_{ln}\in\mathcal{F}_{ln}$), then $\bm z_{ln}^\star=\bm r_{ln}$.
\item[(ii)] If $\mathcal{D}_{ln}\neq\varnothing$, then every global minimizer satisfies
\vspace{-0.5 em}\begin{equation}\label{z_ln}
\boldsymbol{z}_{ln}^{\star} \in \bigcup_{\ell \in \mathcal{D}_{ln}}\left(\mathcal{W}_{l\ell} \cup \mathcal{U}_{l\ell}\right).
\vspace{-0.5 em}\end{equation}
\end{itemize}
\end{lemma}

\begin{IEEEproof}
Suppose $\mathcal{D}_{ln}\neq\varnothing$. If a feasible point $\boldsymbol{z}$ lies in the 
interior of $\mathcal F_{ln}$, then there exists $\delta>0$ such that
$\|\boldsymbol{z}-\boldsymbol{z}_{l\ell}\|_2\ge D+\delta$, $\forall \ell\neq n$. For sufficiently
small $\tau\in(0,1)$, the point $\boldsymbol{z}_\tau=(1-\tau)\boldsymbol{z}+\tau \boldsymbol{r}_{ln}$
remains feasible, and $f(\boldsymbol{z}_\tau)<f(\boldsymbol{z})$, a contradiction. Hence any
minimizer must lie on at least one active boundary $\mathrm{Circle}_{l\ell}$ with
$\ell\in \mathcal{D}_{ln}$.
Fix such an $\ell$. The unconstrained minimizer of $f$ on the circle
$\mathrm{Circle}_{l\ell}$ is the radial projection of $\boldsymbol{r}_{ln}$ onto the circle,
i.e., one of the points in $\mathcal{U}_{l\ell}$ in~\eqref{U_l}.
If this point is feasible with respect to all other constraints, it belongs
to $\mathcal U_{l\ell}$ and is optimal. Otherwise, the feasible subset of
$\mathrm{Circle}_{l\ell}$ is a finite union of closed arcs whose endpoints are
precisely the circle--circle intersections with $\mathrm{Circle}_{l\ell'}$,
$\forall \ell'\neq n,\ell$. Since the closest point on the full circle is infeasible, the minimum of $f$ over each feasible closed arc is attained at an endpoint; hence the minimizer must lie in $\mathcal W_{l\ell}$. Taking the union over all
$\ell\in \mathcal{D}_{ln}$ yields~\eqref{z_ln}.
\end{IEEEproof}

We next show that it suffices to inspect a \emph{finite} candidate set on the relevant circles. Intuitively, along any fixed circle $\mathrm{Circle}_{ln}$ the unconstrained minimizer of $f$ is the orthogonal projection of $\bm r_{ln}$ onto that circle in the radial direction; when this point is excluded by the other inequality constraints, the constrained minimum over the remaining arcs occurs at arc endpoints, which correspond to circle--circle intersection points.
The following two lemmas describe a more direct way of constructing these sets.
\begin{lemma}\label{lem:best-on-circle}
Fix $\ell\in\mathcal{N}\setminus\{n\}$ and minimize $f(\boldsymbol{z})=\left\|\boldsymbol{z}-\boldsymbol{r}_{l n}\right\|_{2}^{2}$ over $\mathrm{Circle}_{l\ell}$. The stationary points are exactly the two intersections of $\mathrm{Circle}_{l\ell}$ with $\mathrm{Line}_{l,\ell n}$, given by
\begin{equation}\label{eq:U-raw}
\widetilde{\mathcal{U}}_{l\ell}\ =\
\begin{cases}
\varnothing, & \bm r_{ln}=\bm z_{l\ell},\\
\bm z_{l\ell} \pm \dfrac{D(\bm r_{ln}-\bm z_{l\ell})}{\|\bm r_{ln}-\bm z_{l\ell}\|_2}, & \text{otherwise}.
\end{cases}
\end{equation}
Among these two points, the nearer one to $\bm r_{ln}$ is $\boldsymbol{z}_{l \ell}+D \frac{\boldsymbol{r}_{l n}-\boldsymbol{z}_{l \ell}}{\left\|\boldsymbol{r}_{l n}-\boldsymbol{z}_{l \ell}\right\|_{2}}$.
\end{lemma}
\begin{IEEEproof}
Parametrize $\bm z=\bm z_{l\ell} + D\bm u$ with $\|\bm u\|_2=1$. Then
\begin{equation}
\|\bm z-\bm r_{ln}\|_2^2=\bigl\|\bm z_{l\ell}-\bm r_{ln}\bigr\|_2^2 + D^2 + 2D\,\bm u^T(\bm z_{l\ell}-\bm r_{ln}),
\end{equation}
which is minimized when $\boldsymbol{u}$ aligns with $-\dfrac{\boldsymbol{z}_{l\ell}-\boldsymbol{r}_{ln}}{\lVert \boldsymbol{z}_{l\ell}-\boldsymbol{r}_{ln}\rVert_2}$, yielding the two radial intersections above, and selecting the branch corresponding to the positive sign yields the intersection nearer to $\boldsymbol{r}_{ln}$.
\end{IEEEproof}
\begin{lemma}[Arc removal by a single inequality]\label{lem:arc-removal}
Fix $\ell \in \mathcal N\setminus\{n\}$ and any $k\in \mathcal N\setminus\{n,\ell\}$, the constraint $\|\boldsymbol{z} - \boldsymbol{z}_{lk}\|_2 \ge D$
removes at most one open arc from $\mathrm{Circle}_{l\ell}$. The endpoints of that arc are exactly the intersections of $\mathrm{Circle}_{l\ell}$ and
$\mathrm{Circle}_{lk}$. Let $d \triangleq \| \boldsymbol{z}_{lk} - \boldsymbol{z}_{l\ell} \|_2$. Then, the set of
endpoints $\mathcal W_{l,\ell k}$ is
\begin{equation}\label{W_l_k}
\widetilde{\mathcal{W}}_{l,\ell k}\! =\!
\begin{cases}
\varnothing, & d = 0,\\
\big\{ \dfrac{\boldsymbol{z}_{l\ell}+ \boldsymbol{z}_{lk}}{2}\big\}, & d = 2D,\\
\dfrac{\boldsymbol{z}_{l\ell}+ \boldsymbol{z}_{lk}}{2} \pm  \dfrac{1}{2}\sqrt{4D^2 - d^2}\hat {\boldsymbol{n}}_{\ell k}, & 0<d<2D,\\
\varnothing, & d>2D,
\end{cases}
\end{equation}
where $\hat{\boldsymbol{n}}_{\ell k}$ is any unit vector orthogonal to $\boldsymbol{z}_{lk}-\boldsymbol{z}_{l\ell}$. For example,
writing $\boldsymbol{z}_{l\ell}=[x_{l\ell}, y_{l\ell}]^T$ and $\boldsymbol{z}_{lk}=[x_{lk},y_{lk}]^T$,
\vspace{-0.5 em}\begin{equation}
\hat{\boldsymbol{n}}_{\ell k}= \frac{1}{\|\boldsymbol{z}_{lk}-\boldsymbol{z}_{l\ell}\|_2}
\begin{bmatrix} y_{lk} - y_{l\ell}\\ -(x_{lk} - x_{l\ell})\end{bmatrix}.
\vspace{-0.5 em}\end{equation}
\end{lemma}
\begin{IEEEproof}
See Appendix~\ref{proof_A}.
\end{IEEEproof}

%
%
\begin{figure}[t]
	\centering\vspace{-1.em}
	\includegraphics[width=0.5\textwidth]{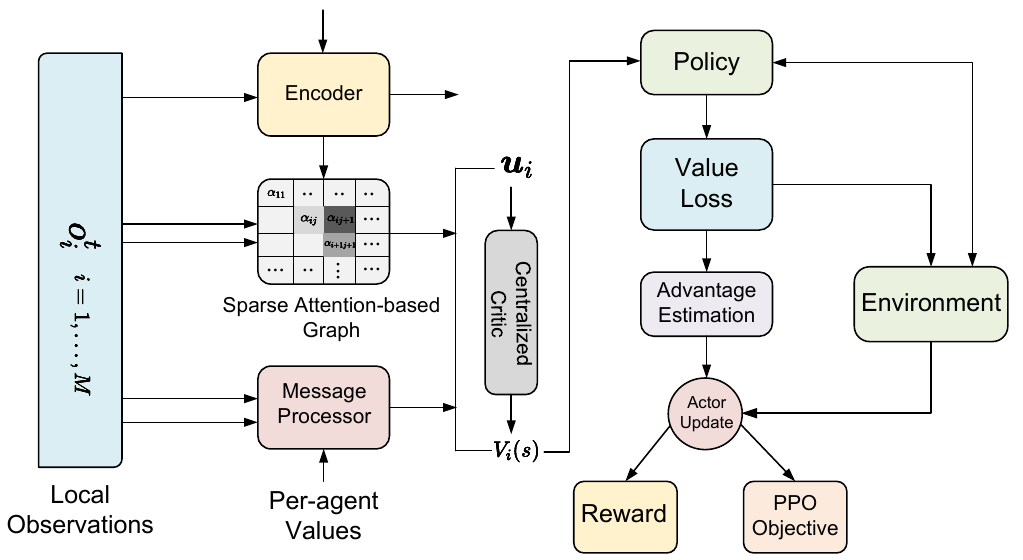}
    \vspace{-2.em}
	\caption{\scopeedit{The heterogeneous MARL architecture}, including a sparse attention-based graph for adaptive agent communication and a centralized critic for consistent advantage estimation across different agent roles.} 
\label{fig:marl_model}
\end{figure}

The auxiliary variable subproblem $(\mathcal{P} \text{1-a})$ can be solved efficiently by exploiting its geometric structure. The optimal solution of this non-convex problem is guaranteed to lie on the boundary of the feasible region, which is defined by the exteriors of multiple exclusion disks. This observation reduces the continuous optimization to a finite search over candidate points, consisting of radial projections and circle--circle intersection points. Having established this geometric mechanism for enforcing the distance-based constraints, we now shift our focus to the second component of the alternating optimization procedure: updating the original system variables \scopeedit{within a heterogeneous multi-agent framework.}
\vspace{-0.45cm}
\subsection{\scopeedit{Heterogeneous Coordination Network}}

The heterogeneous MARL schemes developed in \cite{Zhong2024HARL, 9351684, yu2025heterogeneous} can solve distributed optimization efficiently. However, they do not explicitly capture inter-AP information exchange and do not scale gracefully with increasing agent population, which can lead to prohibitive interaction overhead in large networks. To address this limitation, we propose \scopeedit{a heterogeneous coordination architecture} that integrates a dynamic-interaction graph critic with a LIIR module that learns a distinct intrinsic reward for each AP agent. \scopeedit{Together, these components support coordinated policy learning under centralized training with decentralized online execution.}

The alternating optimization approach for problem $(\mathcal{P}\text{1})$ involves two main steps. In Section~\ref{pro_recons}, we developed a geometric method to solve for the auxiliary variables $\left\{\boldsymbol{z}_{ln}\right\}_{\forall l, n}$ while keeping the original variables fixed. We now address the other half of this procedure, which is the optimization of the original variables: the downlink power allocations $\{\rho_{lk}\}_{\forall l, k}$ and the physical MA positions $\{\boldsymbol{r}_{ln}\}_{\forall l, n}$, while the auxiliary variables $\{\boldsymbol{z}_{ln}\}_{\forall l, n}$ are held constant. The formulation is:
\begin{align}
(\mathcal{P} \text{1-b}): \quad \max _{\rho_{l k}, \boldsymbol{r}_{ln}} &\sum_{k=1}^{K} R_k-\varrho\sum_{l=1}^{L}\sum_{n=1}^{N}\!\left\|\boldsymbol{r}_{ln}-\boldsymbol{z}_{ln}\right\|_{2}^{2}\\
\text { s.t. } &\sum_{k=1}^{K} \rho_{l k} \leq P_{l}^{\max }, \quad \rho_{l k} \geq 0, \quad \forall l, k,\\
&\boldsymbol{r}_{ln} \in \mathcal{C}, \quad \forall l, n.
\end{align}

While subproblem $(\mathcal{P} \text{1-b})$ could be addressed leveraging centralized optimization methods at the CPU, such an approach scales poorly with the number of APs and UEs, creating a significant communication bottleneck for acquiring system-wide state information and disseminating control decisions. \scopeedit{To support decentralized online policy execution after centralized training}, we recast this optimization task within a MARL framework. This paradigm naturally matches the distributed structure of CF systems, allowing APs to learn effective local control policies from local observations while collectively optimizing a global objective. Moreover, MARL includes a range of efficient algorithms, such as MAPPO, in which cooperative agents iteratively interact with the environment, receive reward-based feedback, and update their policies to progressively improve performance.

In this MARL setting, we model decision making at each AP using two role-specialized agent types: a power-control agent responsible for downlink power allocation and a positioning agent responsible for adjusting MA locations. This heterogeneous structure enables task-specific policy learning for the distinct control dimensions. To facilitate efficient and coordinated learning among the resulting large agent population, we further introduce a reward-decoupling mechanism that mitigates credit-assignment ambiguity and stabilizes training.

\begin{revblock}
Each AP $l$ is assigned two role-specialized agents: $(l,\rho)$ for AP-level power allocation and $(l,r)$ for jointly controlling the $N$ MA displacements in $\mathbb{R}^{2N}$. Thus, $\mathcal{M}=\mathcal{M}_{\rho}\cup\mathcal{M}_{r}$ with $\mathcal{M}_{\rho}=\{(l,\rho)\}_{l=1}^{L}$ and $\mathcal{M}_{r}=\{(l,r)\}_{l=1}^{L}$, yielding $M=|\mathcal{M}|=2L$.
\end{revblock}
As shown in Fig.~\ref{fig:marl_model}, we consider a fully cooperative multi-agent MDP with $M$ agents indexed by $i\in\{1,\dots,M\}$.
An MDP is characterized by a tuple $<\mathcal{O}_{i}, \mathcal{S}, \mathcal{A}_{i}, \mathcal{R}, \mathcal{T}, \gamma>$, where $\mathcal{O}_{i}$, $\mathcal{S}$, and $\mathcal{A}_{i}$ represent the observation space, state space, and action space of agent $i$, respectively.
At time $t$, the environment is in global state $\boldsymbol{s}_t\in\mathcal{S}$; agent $i$ observes $\boldsymbol{o}^i_t\in\mathcal{O}_i$ and takes action $\boldsymbol{a}^i_t\in\mathcal{A}_i$.
The joint action is $\boldsymbol{u}_t\coloneqq(\boldsymbol{a}_t^1,\dots,\boldsymbol{a}_t^M)$ and the joint policy factorizes as
\vspace{-0.5 em}\begin{equation}
\pi_{\boldsymbol{\theta}}(\boldsymbol{u}_{t}\mid \boldsymbol{s}_{t})=\prod_{i=1}^M \pi_{\theta_i}(\boldsymbol{a}_t^i\mid \boldsymbol{o}_t^i), \qquad \boldsymbol{\theta}=(\theta_1,\dots,\theta_M),
\vspace{-0.5 em}\end{equation}
with discount factor $\gamma\in(0,1)$. Agents are updated sequentially in \revise{a fixed or sampled global order} per iteration.

\begin{revblock}
For each agent $i$ associated with AP $l$, the local observation is derived from the instantaneous CSI at AP $l$. The local CSI matrix is given by
\begin{equation}
\boldsymbol{H}_{l,t}
=
\left[
\boldsymbol{h}_{l1}(\boldsymbol{r}_{l,t}),
\boldsymbol{h}_{l2}(\boldsymbol{r}_{l,t}),
\ldots,
\boldsymbol{h}_{lK}(\boldsymbol{r}_{l,t})
\right]
\in\mathbb{C}^{N\times K}.
\end{equation}
The complex-valued CSI is represented by stacking its real and imaginary components and applying feature-wise normalization:
\begin{equation}
\boldsymbol{o}_t^i
=
\operatorname{Norm}
\left(
\left[
\operatorname{vec}\left(\Re\{\boldsymbol{H}_{l,t}\}\right),
\operatorname{vec}\left(\Im\{\boldsymbol{H}_{l,t}\}\right)
\right]
\right)\in\mathbb{R}^{2NK},
\end{equation}
where $\operatorname{Norm}(\cdot)$ denotes feature-wise normalization. This observation is shared by the power-control and MA-positioning actors of AP $l$. During training only, the centralized critic may use the aggregated observations $\boldsymbol{s}_t=\{\boldsymbol{o}_t^i\}_{i=1}^{M}$, or equivalently the global CSI collection $\{\boldsymbol{H}_{l,t}\}_{l=1}^{L}$.
\end{revblock}

\begin{revblock}
Both role-specialized action spaces are continuous. For AP $l$, the power-control actor generates the normalized vector $\boldsymbol{a}_{(l,\rho),t}=[a_{l1,t}^{\rho},\ldots,a_{lK,t}^{\rho}]^{T}\in\mathcal{A}_{(l,\rho)}\subset(0,1)^K$, subject to $\sum_{k=1}^{K}a_{lk,t}^{\rho}\le 1$.
The corresponding transmit powers are obtained by scaling the action as $\rho_{lk,t}=P_l^{\max}a_{lk,t}^{\rho}$, $k=1,\ldots,K$, which ensures $\rho_{lk,t}\ge0$ and $\sum_{k=1}^{K}\rho_{lk,t}\le P_l^{\max}$.
For the positioning role, AP $l$ outputs $\boldsymbol{a}_{(l,r),t}\in[-1,1]^{2N}$, where each pair of entries represents the two-dimensional displacement of one MA. The MA positions are updated as
\begin{equation}
\boldsymbol{r}_{ln,t+1}
=
\Pi_{\mathcal C}
\left(
\boldsymbol{r}_{ln,t}
+
\Delta_{\max}
\boldsymbol{a}_{ln,t}^{r}
\right),
\qquad n=1,\ldots,N,
\label{eq:position_action_mapping}
\end{equation}
where $\Pi_{\mathcal C}$ projects onto the feasible movement region $\mathcal C$, and $\Delta_{\max}$ denotes the maximum displacement per step.
\end{revblock}

At time $t$, agent $i$ receives an observation $\boldsymbol{o}^{i}_{t}\in\mathbb{R}^{d_o}$. A shared encoder \(\phi_{\mathrm{enc}}:\mathbb{R}^{d_o}\!\to\!\mathbb{R}^{d_e}\) produces node embeddings
\begin{equation}\label{e_i}
\boldsymbol{e}^i_t=\phi_{\mathrm{enc}}(\boldsymbol{o}^{i}_{t})\in\mathbb{R}^{d_e}.
\end{equation}
From each $\boldsymbol{e}^{i}_{t}$, we compute queries, keys, and values with shared linear maps
\begin{equation}
\boldsymbol{q}^i_t=\boldsymbol{W}_q \boldsymbol{e}^i_t\in\mathbb{R}^{d_k},\quad
\boldsymbol{k}^i_t=\boldsymbol{W}_k \boldsymbol{e}^i_t\in\mathbb{R}^{d_k},\quad
\boldsymbol{v}^i_t=\boldsymbol{W}_v \boldsymbol{e}^i_t\in\mathbb{R}^{d_v},
\end{equation}
where $\boldsymbol{W}_q\in\mathbb{R}^{d_k\times d_e}$, $\boldsymbol{W}_k\in\mathbb{R}^{d_k\times d_e}$, and $\boldsymbol{W}_v\in\mathbb{R}^{d_v\times d_e}$.
Define scaled dot-product compatibility scores with self-edges masked:
\begin{equation}
c_{ij,t}=
\begin{cases}
\dfrac{{\boldsymbol{q}^i_t}^{T} \boldsymbol{k}^j_t}{\sqrt{d_k}}, & j\neq i,\\
-\infty, & j=i.
\end{cases}
\end{equation}
We further introduce a temperature parameter $T>0$ and define $\tilde{c}_{ij,t} = c_{ij,t}/T$.
\begin{revblock}
Let \(\mathcal{N}_{i}^{\rm att}(t)\subset\{1,\dots,M\}\setminus\{i\}\) denote the indices of the $\kappa$ largest entries in $\{\tilde{c}_{ij,t}\}_{j\neq i}$. Pure top-$\kappa$ pruning may disconnect a directed interaction graph when all attention scores associated with an AP/agent are weak. To avoid this issue, we use a connectivity-preserving sparse attention mask. Each node also keeps a fixed low-overhead backbone neighbor set $\mathcal{N}_{i}^{\rm back}$, constructed from a connected nearest-neighbor or ring-type AP backbone and mapped to the corresponding role-specialized agent nodes. The active neighbor set is
\begin{equation}
\mathcal{N}_{i}^{\rm act}(t)
=
\mathcal{N}_{i}^{\rm att}(t)
\cup
\mathcal{N}_{i}^{\rm back}.
\label{eq:connectivity_neighbor_set}
\end{equation}
Since the backbone graph is connected, the union graph $G_t=(\mathcal{V},\mathcal{E}_t)$ with $(i,j)\in\mathcal{E}_t$ iff $j\in\mathcal{N}_{i}^{\rm act}(t)$ remains weakly connected for every $t$. Thus, top-$\kappa$ pruning suppresses redundant non-backbone links but cannot isolate an AP/agent during training.
We form a row-stochastic, backbone-augmented sparse adjacency by normalizing over the active neighbor set:
\begin{equation}\label{alpha_ij}
\alpha_{ij,t}=
\frac{\exp(\tilde{c}_{ij,t})}{\sum_{\ell\in \mathcal{N}_{i}^{\rm act}(t)}\exp(\tilde{c}_{i\ell,t})}\,\boldsymbol{1}\!\left\{j\in\mathcal{N}_{i}^{\rm act}(t)\right\},
\end{equation}
where
$\sum_{j=1}^{M}\alpha_{ij,t}=1$.
The directed message arriving at node \(i\) is the convex combination
\begin{equation}
\boldsymbol{m}^i_t=\sum_{j\in \mathcal{N}_{i}^{\rm act}(t)} \alpha_{ij,t}\,\boldsymbol{v}^j_t\in\mathbb{R}^{d_v}.
\end{equation}
The above attention equations implement a single-layer graph-attention critic on a time-varying, connectivity-preserving sparse directed graph \(G_t=(\mathcal{V},\mathcal{E}_t)\) with \((i,j)\in\mathcal{E}_t\) iff \(j\in\mathcal{N}_{i}^{\rm act}(t)\).
\end{revblock}
A per-node update then fuses local and aggregated features,
\begin{equation}
\boldsymbol{\omega}^i_t=\sigma \bigl(\boldsymbol{W}_m\,[\boldsymbol{e}^i_t; \boldsymbol{m}^i_t]+\boldsymbol{b}_m\bigr)\in\mathbb{R}^{d_u},
\end{equation}
and a scalar value head produces per-agent values and the team value
\begin{equation}
V_i(\boldsymbol{s}_t)=\psi(\boldsymbol{\omega}^i_t),\qquad
V_{\mathrm{ex}}(\boldsymbol{s}_t)=\sum_{i=1}^{M} V_i(\boldsymbol{s}_{t}).
\end{equation}
The extrinsic TD advantage used to update the actors is
\begin{equation}
A_{\mathrm{ex}}(\boldsymbol{s}_t,\boldsymbol{u}_t)= r_{\mathrm{ex},t}+\gamma V_{\mathrm{ex}}(\boldsymbol{s}_{t+1})-V_{\mathrm{ex}}(\boldsymbol{s}_t).
\end{equation}

\begin{proposition}[Permutation invariance]\label{Permutation}
Let $\Pi$ be any permutation of agent indices with permutation matrix $\boldsymbol{P} \in \{0, 1\}^{M \times M}$. If we relabel agents by $\Pi$, equivalently, apply $\boldsymbol{P}$ to the rows of all per-agent quantities, then the critic's scalar output is invariant:
\begin{equation}
    V_{\text{ex}}(\boldsymbol{s}_t) = V_{\text{ex}}(\Pi \boldsymbol{s}_t), \quad \forall t.
\end{equation}
\end{proposition}
\begin{IEEEproof}
See Appendix~\ref{proof_B}.
\end{IEEEproof}
Permutation invariance means the critic treats agents as an unordered set, which is essential in cell-free deployments where agent identities carry no semantics beyond their observations.

Let $r_{\mathrm{ex}}(\boldsymbol{s}_t,\boldsymbol{u}_t)$ denote the per-step \emph{extrinsic} team reward, and let $r_{\mathrm{in}}^i(\boldsymbol{o}_t^i,\boldsymbol{a}_t^i;\eta_i)$ be an \emph{intrinsic} shaping reward produced by agent $i$ with parameters $\eta_i$.
Concretely, following the optimization in $(\mathcal{P} \text{1-b})$, we instantiate
\begin{equation}
\label{eq:extrinsic}
r_{\mathrm{ex}}(\boldsymbol{s}_t,\boldsymbol{u}_t) \triangleq \sum_{k=1}^{K} R_k(\boldsymbol{s}_t,\boldsymbol{u}_t)-\varrho\,\sum_{l=1}^{L}\sum_{n=1}^{N}\left\|\boldsymbol{r}_{ln,t}-\boldsymbol{z}_{ln,t}\right\|_2^2,
\end{equation}
where $R_k(\boldsymbol{s}_t,\boldsymbol{u}_t)$ denotes the downlink rate of user $k$ at time $t$ and $\varrho>0$ is the proximity weight coefficient.
\label{par:reward_scope}
\scopeedit{The extrinsic team reward consists of the nominal sum rate and a geometric feasibility penalty, while practical factors such as movement delay, actuation energy, and CSI/control overhead are not included in the reward formulation.}
\label{par:reward_scope_end}
The proxy reward used when updating agent $i$ mixes extrinsic and its own intrinsic rewards:
\begin{equation}
\label{eq:proxy_reward}
r_{\mathrm{px}}^i(\boldsymbol{s}_t,\boldsymbol{u}_t) \triangleq r_{\mathrm{ex}}(\boldsymbol{s}_t,\boldsymbol{u}_t)\;+\;\lambda\,r_{\mathrm{in}}^i(\boldsymbol{o}_t^i,\boldsymbol{a}_t^i;\eta_i),
\end{equation}
where $\lambda>0$ is a hyper-parameter that balances the extrinsic team reward and the distinct intrinsic reward.

Then, we define the discounted extrinsic return as
\begin{equation}
R_{\mathrm{ex}, t} \triangleq \sum_{\ell=0}^{\infty} \gamma^{\ell} r_{\mathrm{ex}}\left(\boldsymbol{s}_{t+\ell}, \boldsymbol{u}_{t+\ell}\right),
\end{equation}
and, accordingly, the associated state-value function \rev{under the candidate policy $\pi_{\boldsymbol{\theta}}$} as
\begin{equation}
V_{\mathrm{ex}}(\boldsymbol{s}) \triangleq \mathbb{E}_{\pi_{\boldsymbol{\theta}}}\left[R_{\mathrm{ex}, t} \mid \boldsymbol{s}_{t}=\boldsymbol{s}\right].
\end{equation}
By the same token, the discounted proxy return is
\begin{equation}
R_{\mathrm{px}, t}^{i} \triangleq \sum_{\ell=0}^{\infty} \gamma^{\ell} r_{\mathrm{px}}^{i}\left(\boldsymbol{s}_{t+\ell}, \boldsymbol{u}_{t+\ell}\right),
\end{equation}
and the corresponding state-value function is
\begin{equation}
V_{\mathrm{px}}^{i}(\boldsymbol{s}) \triangleq \mathbb{E}_{\pi_{\boldsymbol{\theta}}}\left[R_{\mathrm{px}, t}^{i} \mid \boldsymbol{s}_{t}=\boldsymbol{s}\right].
\end{equation}
For any subset $\EuScript{S}\subseteq\{1,\dots,M\}$, define the $i$-th agent proxy $Q$ marginalized over the complementary agents:
\begin{equation}
\label{eq:marginal_q}
Q_{\mathrm{px},\EuScript{S}}^{\,i}(\boldsymbol{s},\boldsymbol{a}_{\EuScript{S}})
\!=\!\mathbb{E}_{\boldsymbol{a}_{-\EuScript{S}}\sim \pi_{\boldsymbol{\theta},-\EuScript{S}}}\!\left[
\sum_{t\ge 0}\gamma^t r_{\mathrm{px}}^{i}(\boldsymbol{s}_t,\boldsymbol{u}_t)\!\Bigm|\!
\substack{\boldsymbol{s}_0=\boldsymbol{s}\\ \boldsymbol{a}_{\EuScript{S},0}=\boldsymbol{a}_{\EuScript{S}}}
\right],
\end{equation}
where $\boldsymbol{a}_{\EuScript{S}}$ stacks the actions of agents in $\EuScript{S}$ and $-\EuScript{S}$ denotes its complement.
During the sequential update at step $i$, let the prefix set be $\EuScript{P}_i\coloneqq\{1,\dots,i-1\}$.
We define the HAPPO proxy advantage for agent $i$ by the marginal $Q$-difference
\begin{equation}
\label{eq:proxy_advantage}
A_{\mathrm{px}}^{i}(\boldsymbol{s},\boldsymbol{a}_{\EuScript{P}_{i}},\boldsymbol{a}^i)=\!Q_{\mathrm{px},\EuScript{P}_{i}\cup\{i\}}^{i}(\boldsymbol{s},\boldsymbol{a}_{\EuScript{P}_{i}},\boldsymbol{a}^i)-Q_{\mathrm{px}, \EuScript{P}_{i}}^{i}(\boldsymbol{s},\boldsymbol{a}_{\EuScript{P}_{i}}).
\end{equation}
$A_{\mathrm{px}}^{i}$ measures the value gain of selecting $\boldsymbol{a}^i$ for the current agent $i$ given the already-updated prefix actions $\boldsymbol{a}_{\EuScript{P}_{i}}$, relative to marginalizing $\boldsymbol{a}^i$ out.
In practice, we estimate~\eqref{eq:proxy_advantage} with a centralized proxy critic for agent $i$ and GAE($\lambda_{\mathrm{GAE}}$), and we stop gradients through critics in all policy updates.

When updating agent $i$ in iteration $k$, keep the prefix policies fixed at their new parameters $\{\theta_1^{k+1},\dots,\theta_{i-1}^{k+1}\}$ and the rest at their old parameters $\{\theta_i^k,\dots,\theta_M^k\}$.
Define the per-step importance ratio
\begin{equation}
\label{eq:ratio}
r_t^{\,i}(\theta_i)=\frac{\pi_{\theta_i}(\boldsymbol{a}_t^i\mid \boldsymbol{o}_t^i)}{\pi_{\theta_i^{k}}(\boldsymbol{a}_t^i\mid \boldsymbol{o}_t^i)}.
\end{equation}
The HAPPO objective for agent $i$ is
\begin{equation}
\label{eq:happo_obj}
\max_{\theta_i}\;
\mathbb{E}\!\left[
\begin{aligned}
&\min\bigl(r_t^{i}(\theta_i)\,A_{\mathrm{px}}^{i},
             \mathrm{clip}_{[1-\epsilon,1+\epsilon]}\bigl(r_t^{i}(\theta_i)\bigr)A_{\mathrm{px}}^{i}\bigr)\\
&+\beta\mathcal{H}\bigl(\pi_{\theta_i}(\cdot\mid \boldsymbol{o}_t^i)\bigr)
\end{aligned}
\right],
\end{equation}
where $\epsilon>0$ is the PPO clipping range, $\beta\ge 0$ an entropy weight, and $\mathcal{H}$ the Shannon entropy.                               
We typically perform a few epochs of stochastic gradient ascent on~\eqref{eq:happo_obj} using minibatches from the current rollout, with advantage whitening and per-agent KL monitoring.
\textcolor{black}{\eqref{eq:happo_obj} is the PPO-Clip update used in the implementation and in Algorithm~\ref{alg:marl_scalable}. \eqref{eq:trust_region} is the KL-constrained trust-region form used to state the monotonic-improvement guarantee in Theorem~\ref{thm:bias_aware}. Thus, Theorem~\ref{thm:bias_aware} directly justifies the trust-region formulation. It also applies to the implemented PPO-Clip update when the realized clipped step is KL-feasible, i.e., when the monitored average KL of each agent is below the corresponding budget $\delta_i$. Without this KL-feasibility condition, \eqref{eq:happo_obj} should be interpreted as a practical surrogate for the trust-region step, rather than as a separate monotonic-improvement theorem.}
To make the theoretical KL control explicit, one may write the idealized trust-region problem as
\begin{equation}
\label{eq:trust_region}
\begin{aligned}
\max_{\theta_i} \quad & \mathbb{E}\!\left[r_t^{\,i}(\theta_i)\,A_{\mathrm{px}}^{\,i}\right] \\
\text{s.t.} \quad & \mathbb{E}\!\left[D_{\mathrm{KL}}\!\left(\pi_{\theta_i^{k}}(\cdot\mid \boldsymbol{o}_t^i)\,\big\|\,\pi_{\theta_i}(\cdot\mid \boldsymbol{o}_t^i)\right)\right] \le \delta,
\end{aligned}
\end{equation}
solved by a conjugate-gradient step with backtracking.
\textcolor{black}{In the numerical implementation, we optimize the clipped objective in \eqref{eq:happo_obj} and use the KL budget associated with \eqref{eq:trust_region} as an acceptance/early-stopping criterion. This makes the implemented PPO-Clip step a KL-controlled approximation to the trust-region update analyzed below.}
Let $\boldsymbol{\eta}\coloneqq(\eta_1,\dots,\eta_M)$ be the intrinsic-reward parameters and $\boldsymbol{\theta}'(\boldsymbol{\eta})$ the policy parameters after one inner HAPPO pass using proxy advantages. The LIIR outer objective is the extrinsic return under the post-update policies:
\begin{equation}
\label{eq:outer_obj}
J_{\mathrm{ex}}(\pi_{\boldsymbol{\theta}'(\boldsymbol{\eta})})=\mathbb{E}_{\pi_{\boldsymbol{\theta}'(\boldsymbol{\eta})}}\!\Big[\sum\nolimits_{t\ge 0}\gamma^t\,r_{\mathrm{ex}}(\boldsymbol{s}_t,\boldsymbol{u}_t)\Big],
\end{equation}
leading to the bilevel program
\begin{equation}
\label{eq:bilevel}
\begin{aligned}
\max_{\boldsymbol{\eta}} \quad & J_{\mathrm{ex}}(\pi_{\boldsymbol{\theta}^\star(\boldsymbol{\eta})}) \\
\text{s.t.} \quad & \boldsymbol{\theta}^\star(\boldsymbol{\eta})\in\arg\max_{\boldsymbol{\theta}\in\mathcal{U}}\,
                 \mathcal{L}_{\mathrm{HAPPO}}(\boldsymbol{\theta};\boldsymbol{\eta}).
\end{aligned}
\end{equation}
where $\mathcal{U}$ encodes the feasible update steps and $\mathcal{L}_{\mathrm{HAPPO}}$ abbreviates the inner-loop objective \eqref{eq:trust_region}.
\textcolor{black}{For the monotonic-improvement analysis, feasibility means satisfying the per-agent average-KL budgets in \eqref{eq:trust_region}; a PPO-Clip iterate generated by \eqref{eq:happo_obj} is included in $\mathcal{U}$ only when it passes this KL check.}
Applying the chain rule for a single inner update yields the LIIR meta-gradient for agent $i$:
\begin{equation}
\label{eq:chain_rule}
\nabla_{\eta_i} J_{\mathrm{ex}}(\pi_{\boldsymbol{\theta}'(\boldsymbol{\eta})})=\nabla_{\theta_i'} J_{\mathrm{ex}}(\pi_{\boldsymbol{\theta}'})\cdot\nabla_{\eta_i}\theta_i'.
\end{equation}
Using a centralized extrinsic critic, we estimate 
\begin{equation}
\label{eq:outer_pg}
\nabla_{\theta_i'} J_{\mathrm{ex}}(\pi_{\boldsymbol{\theta}'})\;\approx\;\mathbb{E}\!\left[\nabla_{\theta_i'}\log \pi_{\theta_i'}(\boldsymbol{a}^i\mid \boldsymbol{o}^i)\;A_{\mathrm{ex}}(\boldsymbol{s},\boldsymbol{u})\right],
\end{equation}
where $A_{\mathrm{ex}}(\boldsymbol{s},\boldsymbol{u})$ denotes the extrinsic advantage computed from the centralized value function as
\begin{equation}
A_{\mathrm{ex}}(\boldsymbol{s},\boldsymbol{u})= r_{\mathrm{ex}}(\boldsymbol{s},\boldsymbol{u})+\gamma V_{\mathrm{ex}}(\boldsymbol{s}')-V_{\mathrm{ex}}(\boldsymbol{s}).
\end{equation}
For one PPO step with learning rate $\alpha_i$, and under the standard LIIR approximation that ignores critic dependence on $\eta_i$, we have
\begin{equation}
\label{eq:inner_sensitivity}
\nabla_{\eta_i}\theta_i'\!\approx\!\alpha_i\,\lambda\mathbb{E}\!\left[\nabla_{\theta_i}\log \pi_{\theta_i}(\boldsymbol{a}^i\mid \boldsymbol{o}^i)\nabla_{\eta_i} r_{\mathrm{in}}^{i}(\boldsymbol{o}^i,\boldsymbol{a}^i;\eta_i)\right].
\end{equation}
\begin{revblock}
Combining the outer policy gradient, the inner sensitivity, and the chain rule gives an unbiased on-policy meta-gradient only under ideal conditions: fresh trajectories from $\pi_{\boldsymbol{\theta}'(\boldsymbol{\eta})}$, exact value and advantage estimates, fully accounted critic dependence on $\eta_i$, and no PPO clipping bias. In GLIIR-HAPPO, these conditions are relaxed due to rollout reuse with clipped ratios, learned centralized/graph critics, and stopped critic-gradient paths through $\eta_i$. The implemented estimator is therefore written as
\begin{equation}
\label{eq:meta_gradient_bias}
\widehat{g}_{\eta_i}^{\mathrm{meta}}
=
\nabla_{\theta_i'}J_{\mathrm{ex}}(\pi_{\boldsymbol{\theta}'})
\cdot
\nabla_{\eta_i}\theta_i'
+
\varepsilon_{\mathrm{meta},i},
\end{equation}
where $\varepsilon_{\mathrm{meta},i}$ aggregates the resulting bias. It vanishes under the ideal conditions above; otherwise, the update is a practical biased meta-gradient estimator.
\end{revblock}

\begin{revblock}
To reduce policy-synchronization overhead, GLIIR-HAPPO can employ role-conditioned federated distillation (FD) \cite{pmlr-v139-zhu21b, sattler2020communication, wu2022communication}. Instead of exchanging full actor parameters or raw observations, agents transmit only actor-output statistics on a small anchor batch. Since power-control and positioning policies differ in action semantics and dimensions, distillation is performed within each role $g\in\{\mathrm{pow},\mathrm{pos}\}$, with agent set $\mathcal{M}_{g}$ and action dimension $d_g$ (\revise{$d_{\mathrm{pow}}=K$} and $d_{\mathrm{pos}}=2N$). The role-specific anchor batch is denoted by $\mathcal{D}_{g}^{\mathrm{FD}}=\{\bar{\boldsymbol{o}}_{b}^{g}\}_{b=1}^{B_{\mathrm{FD}}}$.

For each anchor observation, agent $i\in\mathcal{M}_g$ uploads the Gaussian actor statistics
\begin{equation}
\pi_{\theta_i}(\cdot\mid\bar{\boldsymbol{o}}_{b}^{g})=\mathcal{N}\left(\boldsymbol{\mu}_{i,b}^{g},\operatorname{diag}((\boldsymbol{\sigma}_{i,b}^{g})^2)\right),
\end{equation}
namely $\boldsymbol{\mu}_{i,b}^{g}$ and $\log\boldsymbol{\sigma}_{i,b}^{g}$. The role teacher is obtained by moment averaging,
\begin{equation}
\bar{\boldsymbol{\mu}}_{b}^{g}=\frac{1}{|\mathcal{M}_g|}\sum_{i\in\mathcal{M}_g}\boldsymbol{\mu}_{i,b}^{g},\quad
\log\bar{\boldsymbol{\sigma}}_{b}^{g}=\frac{1}{|\mathcal{M}_g|}\sum_{i\in\mathcal{M}_g}\log\boldsymbol{\sigma}_{i,b}^{g},
\end{equation}
which defines the softened teacher $\bar\pi_g^{\tau_{\mathrm{FD}}}=\mathcal{N}(\bar{\boldsymbol{\mu}}_{b}^{g},\tau_{\mathrm{FD}}^{2}\operatorname{diag}((\bar{\boldsymbol{\sigma}}_{b}^{g})^2))$, where $\tau_{\mathrm{FD}}$ is the FD temperature. Each student minimizes
\begin{equation}
\mathcal{L}_{i}^{\mathrm{FD}}=\frac{1}{B_{\mathrm{FD}}}\sum_{b=1}^{B_{\mathrm{FD}}}D_{\mathrm{KL}}\left(\pi_{\theta_i}^{\tau_{\mathrm{FD}}}(\cdot\mid\bar{\boldsymbol{o}}_{b}^{g})\,\middle\|\,\bar\pi_g^{\tau_{\mathrm{FD}}}(\cdot\mid\bar{\boldsymbol{o}}_{b}^{g})\right).
\end{equation}
The actor update for $i\in\mathcal{M}_g$ becomes
\begin{equation}
\label{eq:fd_actor_update}
\max_{\theta_i}\;\mathcal{J}_{i}^{\mathrm{PPO-FD}}=\mathcal{J}_{i}^{\mathrm{HAPPO}}-\omega_{\mathrm{FD}}\mathcal{L}_{i}^{\mathrm{FD}},
\end{equation}
where $\omega_{\mathrm{FD}}$ is the FD regularization weight and $\mathcal{J}_{i}^{\mathrm{HAPPO}}$ is the clipped objective in \eqref{eq:happo_obj}. Since only $(\boldsymbol{\mu},\log\boldsymbol{\sigma})$ are exchanged, the round-trip overhead is
\begin{equation}
C_{\mathrm{FD}}=4B_{\mathrm{FD}}\sum_{g\in\{\mathrm{pow},\mathrm{pos}\}}|\mathcal{M}_g|d_g\quad\text{floats},
\end{equation}
where the factor four accounts for uploading and broadcasting both statistics. Thus, FD scales with compact action dimensions rather than actor parameter counts, while cross-role coordination is retained through the shared extrinsic reward and dynamic graph critic.
\end{revblock}

Correspondingly, the overall process of the proposed framework is shown in Algorithm~\ref{alg:marl_scalable}.
\begin{algorithm}[t]
\caption{Proposed Multi-Agent Collaboration GLIIR-HAPPO Algorithm}
\label{alg:marl_scalable}

\DontPrintSemicolon

\KwIn{Global environment state $s_0$; max episodes $E_{\max}$, max steps $T$; \revise{penalty schedule $(\varrho_0,\eta_{\varrho},\varrho_{\max},\epsilon_{\rm feas})$; FD parameters $(B_{\mathrm{FD}},\omega_{\mathrm{FD}},\tau_{\mathrm{FD}})$.}}
\KwOut{Optimized actor policies $\{\pi_{\theta_i}^\star\}$; optimal antenna positions $\boldsymbol{r}^\star$ and power allocation $\boldsymbol{\rho}^\star$.}

\textbf{Initialization:} Randomly initialize parameters $\{\theta_i\}$, $\phi$, and $\{\eta_i\}$. \revise{Set $\varrho\leftarrow\varrho_0$ and construct role-specific FD anchor batches $\mathcal{D}_{\mathrm{pow}}^{\mathrm{FD}}$ and $\mathcal{D}_{\mathrm{pos}}^{\mathrm{FD}}$.}\;

\For{$k = 1$ \KwTo $E_{\max}$}{
    Initialize rollout buffer $\mathcal{B} \leftarrow \emptyset$.\;
    \For{$t = 1$ \KwTo $T$}{
        \textbf{Geometric solver:} Update auxiliary variables $\boldsymbol{z}_{ln,t}$ by solving \revise{subproblem} $(\mathcal{P}\text{1-a})$.\;
        \textbf{Observe:} Each agent $i$ receives observation $\boldsymbol{o}_t^i$.\;
        \textbf{Act:} Execute joint action $\boldsymbol{u}_t$.\;
        \textbf{Reward:} Receive extrinsic reward $r_{\text{ex},t}$.\;
        Store transition $(\boldsymbol{s}_t, \boldsymbol{o}_t, \boldsymbol{u}_t, r_{\text{ex},t}, \boldsymbol{z}_{ln,t})$ into $\mathcal{B}$.\;
    }

    Compute intrinsic rewards $r^{\text{in}}_{i,t} \leftarrow r^{\text{in}}_{\eta_i}(\boldsymbol{o}_t^i, \boldsymbol{a}_t^i)$ for all $t \in \mathcal{B}$.\;
    Form proxy rewards $r^{\text{proxy}}_{i,t} \leftarrow r_{\text{ex},t} + \lambda r^{\text{in}}_{i,t}$.\;
    Estimate advantages $\hat{A}_{\text{ex}}$ using $r_{\text{ex}}$ and $\hat{A}_{\text{px}}^i$ leveraging $r^{\text{proxy}}_i$ via GAE and critic $V_\phi$.\;

    Compute meta-gradient $\nabla_{\eta_i} J_{\text{ex}}$ via chain rule (Eq.~\eqref{eq:inner_sensitivity}) adopting $\hat{A}_{\text{ex}}$.\;
    Update intrinsic-reward parameters $\eta_i$.\;

    Generate a random sequential update order $\mathcal{P}$ for agents.\;
    \For{each agent $i$ in order $\mathcal{P}$}{
        Compute prefix-marginalized proxy advantage $A_{\text{px}}^i(\boldsymbol{s}, \boldsymbol{a}_{1:i})$ (Eq.~\eqref{eq:proxy_advantage}).\;
        Update $\theta_i$ by maximizing the PPO-Clip objective with $A_{\text{px}}^i$.\;
        Update importance-sampling weights for subsequent agents.\;
    }

     \revise{For each role $g\in\{\mathrm{pow},\mathrm{pos}\}$, perform role-conditioned FD using $\mathcal{D}_{g}^{\mathrm{FD}}$, construct the teacher distribution, and update each same-role actor by maximizing $\mathcal{J}_{i}^{\mathrm{PPO-FD}}$.}\;
  \revise{Compute the average violation $\nu_k$ using \eqref{eq:penalty_violation} and update $\varrho$ according to the adaptive rule in \eqref{eq:penalty_update}.}\;
}
\end{algorithm}

\vspace{-0.2cm}
\section{Theoretical Analysis of the Proposed Framework}
\label{subsec:happo_liir_theory}
\vspace{-0.05cm}
We derive (i) a HAPPO-specific policy-gradient identity that explicitly accounts for prefix marginalization, (ii) a monotonic improvement bound on the extrinsic objective, enforced via per-agent trust regions, and (iii) an analysis of the impact of replacing purely extrinsic advantages with a proxy advantage that blends intrinsic and extrinsic value signals. Additionally, we establish sufficient conditions under which the proposed LIIR procedure admits two-timescale convergence to stationary points of the \revise{bilevel} objective.

\begin{table}[t]
\caption{\revise{Summary of Symbols and Policy Definitions.}}
\label{tab:symbol_policy_summary}
\centering
{\color{black}
\fontsize{9}{11}\selectfont
\renewcommand{\arraystretch}{1.08}
\begin{tabularx}{\columnwidth}{@{}p{0.36\columnwidth}X@{}}
\toprule
Symbols & Definitions \\
\midrule
$\boldsymbol{s}_t,\boldsymbol{o}_t^i,\boldsymbol{a}_t^i,\boldsymbol{u}_t$ & Global state, local observation, local action, and joint action. \\
$\pi_{\theta_i}$, $\pi_{\boldsymbol{\theta}}$ & Agent-$i$ actor and factored candidate joint policy. \\
$\pi_{\boldsymbol{\theta}^{k}}$, $\pi_{\theta_i^{k}}$ & Joint and per-agent rollout/reference policies at iteration $k$; these replace the ambiguous reference-policy notation $\pi_0$. \\
$\pi_{\boldsymbol{\theta}}$ & Current candidate joint policy used to define the value, $Q$, and advantage functions before the inner update is applied. \\
$\EuScript{P}_{i}$ & Prefix set of already-updated agents. \\
$\pi_{\boldsymbol{\theta}'}$, $\pi_{\boldsymbol{\theta}'(\boldsymbol{\eta})}$ & Post-update joint policy after one HAPPO pass and its dependence on $\boldsymbol{\eta}$. \\
$r_{\mathrm{ex}}$, $r_{\mathrm{in}}^i$, $r_{\mathrm{px}}^i$ & Extrinsic, intrinsic, and proxy rewards. \\
$V_{\mathrm{ex}}^{\pi}(\boldsymbol{s})$, $V_{\mathrm{px}}^{i,\pi}(\boldsymbol{s})$ & Extrinsic and agent-$i$ proxy value functions evaluated under policy $\pi$. \\
$Q_{\mathrm{px},\EuScript{S}}^{\,i,\pi}$, $A_{\mathrm{ex}}^{\pi}$, $A_{\mathrm{px}}^{i,\pi}$ & Marginal proxy $Q$-function and advantage terms evaluated under policy $\pi$. \\
$\boldsymbol{\eta}$, $\lambda$, $\gamma$ & Intrinsic-reward parameters, reward weight, and discount factor. \\
\bottomrule
\end{tabularx}}
\end{table}
\begin{revblock}
For readability, Table~\ref{tab:symbol_policy_summary} summarizes the main symbols and policy scopes used in the subsequent theoretical derivations.
\end{revblock}

Recall the extrinsic reward and proxy reward from \eqref{eq:extrinsic}--\eqref{eq:proxy_reward}, and the HAPPO proxy advantage $A_{\mathrm{px}}^{\,i}$ in \eqref{eq:proxy_advantage}. Let $\rho_{\pi_{\boldsymbol{\theta}}}$ denote the $\gamma$-discounted occupancy measure and $d_{\pi_{\boldsymbol{\theta}}}(\boldsymbol{s})=(1-\gamma)\sum_{t\ge 0}\gamma^t \Pr_{\pi_{\boldsymbol{\theta}}}(\boldsymbol{s}_t=\boldsymbol{s})$ the normalized discounted state visitation. We define the extrinsic multi-agent advantage in the same prefix form as \eqref{eq:proxy_advantage} but exploiting $r_{\mathrm{ex}}$:
\begin{equation}
\label{eq:extrinsic_prefix_adv}
A_{\mathrm{ex}}^{\,i}(\boldsymbol{s},\boldsymbol{a}_{\EuScript{P}_{\!i}},\boldsymbol{a}^i)= Q_{\mathrm{ex},\EuScript{P}_{\!i}\cup\{i\}}^{\,i}(\boldsymbol{s},\boldsymbol{a}_{\EuScript{P}_{i}},\boldsymbol{a}^i)-Q_{\mathrm{ex},\EuScript{P}_{i}}^{i}(\boldsymbol{s},\boldsymbol{a}_{\EuScript{P}_{i}}),
\end{equation}
where $Q_{\mathrm{ex},\EuScript{S}}^{i}$ is defined as in \eqref{eq:marginal_q} with $r_{\mathrm{ex}}$ in place of $r_{\mathrm{px}}^{i}$, i.e.,
\begin{equation}
\begin{aligned}
Q_{\mathrm{ex},\EuScript{S}}^{i}(\boldsymbol{s},\boldsymbol{a}_{\EuScript{S}})
=\ &\mathbb{E}_{\boldsymbol{a}_{-\EuScript{S}}\sim\pi_{\boldsymbol{\theta},-\EuScript{S}}}\Biggl[
\sum_{t\ge0}\gamma^t\, r_{\mathrm{ex}}(\boldsymbol{s}_t,\boldsymbol{u}_t)
\Bigm|
\substack{\boldsymbol{s}_0=\boldsymbol{s}\\ \boldsymbol{a}_{\EuScript{S},0}=\boldsymbol{a}_{\EuScript{S}}}
\Biggr].
\end{aligned}
\end{equation}
For completeness, define the joint extrinsic $Q$-function
\begin{equation}
Q_{\mathrm{ex}}(\boldsymbol{s},\boldsymbol{u})
= \mathbb{E}_{\pi}\!\left[\sum_{t\ge0}\gamma^t\, r_{\mathrm{ex}}(\boldsymbol{s}_t,\boldsymbol{u}_t)\ \middle|\ \boldsymbol{s}_0=\boldsymbol{s},\ \boldsymbol{u}_0=\boldsymbol{u} \right],
\end{equation}
and the joint extrinsic advantage
$A_{\mathrm{ex}}(\boldsymbol{s},\boldsymbol{u})=Q_{\mathrm{ex}}(\boldsymbol{s},\boldsymbol{u})-V_{\mathrm{ex}}(\boldsymbol{s})$.

\begin{assumption}[\revise{Controlled graph-critic and FD perturbations}]
  \label{ass:gfd}
  \begin{revblock}
  At iteration $k$, let $G_k$ be the sparse dynamic-interaction graph selected by the graph critic and let $\widehat A_{\mathrm{px},G_k}^{i}$ denote the proxy advantage estimated from this critic. We assume that the graph-induced approximation and non-stationarity error is uniformly bounded, i.e.,
  \begin{equation}
  \label{eq:graph_error_bound}
  \big|\widehat A_{\mathrm{px},G_k}^{i}(\boldsymbol{s},\boldsymbol{a}_{\EuScript{P}_i},\boldsymbol{a}^i)-A_{\mathrm{px}}^{i}(\boldsymbol{s},\boldsymbol{a}_{\EuScript{P}_i},\boldsymbol{a}^i)\big|\le \varepsilon_{G,k}.
  \end{equation}
  Moreover, the graph sequence is slowly varying in the sense that
  \begin{equation}
  \label{eq:graph_slow_variation}
  \|V_{\phi_{k+1},G_{k+1}}-V_{\phi_k,G_k}\|_{\infty}\le L_G\big(\|\phi_{k+1}-\phi_k\|+\Delta_{G,k}\big),
  \end{equation}
  where $\Delta_{G,k}=M^{-1}\|\boldsymbol{A}_{G_{k+1}}-\boldsymbol{A}_{G_k}\|_1$ measures the attention-graph variation. FD changes the actor direction from the pure HAPPO direction to
  \begin{equation}
  \label{eq:gfd_direction}
  g_{i,k}^{\mathrm{GFD}}=\widehat g_{i,k}^{\mathrm{HAPPO}}-\omega_{\mathrm{FD},k}\nabla_{\theta_i}\mathcal{L}_{i}^{\mathrm{FD}},
  \end{equation}
  with $\|\nabla_{\theta_i}\mathcal{L}_{i}^{\mathrm{FD}}\|\le G_{\mathrm{FD}}$. The realized step is accepted only when it satisfies the per-agent average-KL budget.
  \end{revblock}
  \end{assumption}

\begin{lemma}[Prefix-marginalized policy-gradient identity]
\label{lem:prefix_pg}
For any agent $i$ and parameters $\boldsymbol{\theta}$, define the proxy return $J_{\mathrm{px}}^{i}(\boldsymbol{\theta})\!=\!\mathbb{E}_{\pi_{\boldsymbol{\theta}}}\!\left[\sum_{t\ge 0}\gamma^t\,r_{\mathrm{px}}^{i}(\boldsymbol{s}_t,\boldsymbol{u}_t)\right]$. Then
\begin{equation}
\label{eq:prefix_pg_identity}
\begin{aligned}
\nabla_{\boldsymbol{\theta}_i} J_{\mathrm{px}}^{i}(\boldsymbol{\theta})
&=
\mathbb{E}_{\boldsymbol{s}\sim d_{\pi_{\boldsymbol{\theta}}}}
\ \mathbb{E}_{\substack{
\boldsymbol{a}_{\EuScript{P}_i}\sim\pi_{\boldsymbol{\theta},\EuScript{P}_i} \\
\boldsymbol{a}^i\sim\pi_{\theta_i}
}}\!\Big[
\nabla_{\theta_i}\log \pi_{\theta_i}(\boldsymbol{a}^i \mid \boldsymbol{o}^i)\, \\
&\qquad\qquad
A_{\mathrm{px}}^{i}(\boldsymbol{s},\boldsymbol{a}_{\EuScript{P}_{i}},\boldsymbol{a}^i)
\Big].
\end{aligned}
\end{equation}
\end{lemma}

\begin{IEEEproof}
Express the joint log-likelihood gradient as $\nabla_{\boldsymbol{\theta}}\log\pi_{\boldsymbol{\theta}}(\boldsymbol{u}\!\mid\!\boldsymbol{s})=\sum_{i=1}^M \nabla_{\theta_i}\log \pi_{\theta_i}(\boldsymbol{a}^i\!\mid\!\boldsymbol{o}^i)$. Applying the policy-gradient theorem to $J_{\mathrm{px}}^{i}$ gives
\begin{equation}
\nabla_{\theta_i} J_{\mathrm{px}}^{i}(\boldsymbol{\theta})
=\mathbb{E}_{\pi_{\boldsymbol{\theta}}}\!\left[\nabla_{\theta_i}\log \pi_{\theta_i}(\boldsymbol{a}^i\!\mid\!\boldsymbol{o}^i)\,Q_{\mathrm{px}}^{i}(\boldsymbol{s},\boldsymbol{u})\right].
\end{equation}
Taking conditional expectation over $\boldsymbol{a}_{-\EuScript{P}_i\setminus\{i\}}$ first, and then subtracting the corresponding baseline $Q_{\mathrm{px},\EuScript{P}_i}^{i}(\boldsymbol{s},\boldsymbol{a}_{\EuScript{P}_i})$ (which has zero correlation with $\nabla_{\theta_i}\log \pi_{\theta_i}$) yields \eqref{eq:prefix_pg_identity}. The expectations reduce exactly to the marginalized $Q$-difference in \eqref{eq:proxy_advantage}.
\end{IEEEproof}

\begin{lemma}[Telescoping decomposition of the joint advantage]
\label{lem:telescoping}
For every $(\boldsymbol{s},\boldsymbol{u})$,
\begin{equation}
\label{eq:telescoping}
\vspace{-6pt}
A_{\mathrm{ex}}(\boldsymbol{s},\boldsymbol{u})=\sum_{i=1}^{M} A_{\mathrm{ex}}^{i}(\boldsymbol{s},\boldsymbol{a}_{\EuScript{P}_{i}},\boldsymbol{a}^i).
\vspace{-3pt}
\end{equation}
This identity holds under the mild regularity conditions ensuring the value functions are well-defined.
\end{lemma}
\begin{IEEEproof}
See Appendix~\ref{proof_C}.
\end{IEEEproof}

\begin{theorem}[Joint KL decomposition for factorized multi-agent policies]
\label{KL_decomposition}
Fix a state $\boldsymbol{s}$ and consider the joint action \(\boldsymbol{u}=(\boldsymbol{a}_1,\ldots,\boldsymbol{a}_M)\).
Let the two joint policies parameterized by $\boldsymbol{\theta}$ and $\boldsymbol{\theta}'$ factorize across agents as
\begin{equation}
\pi_{\boldsymbol{\theta}}\!(\boldsymbol{u}\mid \boldsymbol{s})\!=\!\prod_{i=1}^M \!\pi_{\theta_i}(\boldsymbol{a}_i\mid \boldsymbol{o}_i),\!
\pi_{\boldsymbol{\theta}'}\!(\boldsymbol{u}\mid \boldsymbol{s})\!=\!\prod_{i=1}^M \!\pi_{\theta'_i}(\boldsymbol{a}_i\mid \boldsymbol{o}_i),
\end{equation}
where $\boldsymbol{o}_i=\boldsymbol{o}_i(\boldsymbol{s})$ is agent \(i\)'s local observation at state $\boldsymbol{s}$.
Assume that for every \(i\), \(\pi_{\theta_i}(\cdot\mid \boldsymbol{o}_i)\) is absolutely continuous with respect to
\(\pi_{\theta'_i}(\cdot\mid \boldsymbol{o}_i)\) (equivalently, both admit densities w.r.t. a common base measure on \(\mathcal A_i\)).
Then the conditional joint KL divergence decomposes additively:
\begin{equation}\label{eq:KL-sum}
\begin{aligned}
&D_{\mathrm{KL}}\Big(
  \prod_{i=1}^M \pi_{\theta_i}(\cdot\mid \boldsymbol{o}_i)\,\Big\|\,
 \prod_{i=1}^M \pi_{\theta'_i}(\cdot\mid \boldsymbol{o}_i)
\Big)\\
=&
\sum_{i=1}^M D_{\mathrm{KL}}\big(
  \pi_{\theta_i}(\cdot\mid \boldsymbol{o}_i)\,\big\|\,
  \pi_{\theta'_i}(\cdot\mid \boldsymbol{o}_i)
\big).
\end{aligned}
\end{equation}
Consequently, for any distribution $d(\boldsymbol{s})$ over states,
linearity of expectation yields
\begin{equation}\label{eq:avg-KL}
\begin{aligned}
&\mathbb{E}_{\boldsymbol{s}\sim d}\Big[
    D_{\mathrm{KL}}\big(\pi_{\boldsymbol{\theta}}(\cdot\mid \boldsymbol{s})\,\big\|\,\pi_{\boldsymbol{\theta}'}(\cdot\mid \boldsymbol{s})\big)
\Big]\\
=&
\sum_{i=1}^M 
\mathbb{E}_{\boldsymbol{s}\sim d}\Big[
    D_{\mathrm{KL}}\big(\pi_{\theta_i}(\cdot\mid \boldsymbol{o}_i(\boldsymbol{s}))\,\big\|\,\pi_{\theta'_i}(\cdot\mid \boldsymbol{o}_i(\boldsymbol{s}))\big)
\Big].
\end{aligned}
\end{equation}
\end{theorem}
\begin{IEEEproof}
See Appendix~\ref{proof_D}.
\end{IEEEproof}

\begin{lemma}[Proxy--extrinsic advantage gap]
\label{lem:gap_bound}
For all $(\boldsymbol{s},\boldsymbol{a}_{\EuScript{P}_i},\boldsymbol{a}^i)$,
\begin{equation}
\label{eq:gap_bound}
\big|A_{\mathrm{px}}^{i}(\boldsymbol{s},\boldsymbol{a}_{\EuScript{P}_i},\boldsymbol{a}^i)-A_{\mathrm{ex}}^{i}(\boldsymbol{s},\boldsymbol{a}_{\EuScript{P}_i},\boldsymbol{a}^i)\big|
\ \le\ \frac{\lambda\,\bar r_{\mathrm{in}}}{1-\gamma}= \Delta_{\max}.
\end{equation}
\end{lemma}

\begin{IEEEproof}
Both advantages are discounted sums of per-step rewards that differ only in the additive intrinsic term $\lambda r_{\mathrm{in}}^i$ used inside the agent-$i$ proxy return. The difference of the corresponding $Q$-functions is bounded by $\sum_{t\ge 0}\gamma^t \lambda|r_{\mathrm{in}}^i|\le \frac{\lambda\bar r_{\mathrm{in}}}{1-\gamma}$, and baselines cancel in the difference of $Q$'s.
\end{IEEEproof}

\begin{proposition}[Potential-based intrinsic shaping yields no bias]
  \label{prop:potential}
  \begin{revblock}
  If $r_{\mathrm{in}}^i(\boldsymbol{o}^i,\boldsymbol{a}^i;\eta_i)=\gamma\Phi_i(\boldsymbol{s}')-\Phi_i(\boldsymbol{s})$ for some bounded potential $\Phi_i$, then, for any fixed prefix $\boldsymbol{a}_{\EuScript{P}_i}$,
  \begin{equation}
  \label{eq:potential_no_bias}
  A_{\mathrm{px}}^{i}(\boldsymbol{s},\boldsymbol{a}_{\EuScript{P}_i},\boldsymbol{a}^i)=A_{\mathrm{ex}}^{i}(\boldsymbol{s},\boldsymbol{a}_{\EuScript{P}_i},\boldsymbol{a}^i).
  \end{equation}
  Thus, potential-based shaping contributes only an action-independent baseline and does not bias the agent-$i$ policy update.
  \end{revblock}
  \end{proposition}
  
  \begin{IEEEproof}
  \begin{revblock}
  Let $\EuScript{S}\in\{\EuScript{P}_i\cup\{i\},\EuScript{P}_i\}$. For a trajectory with $\boldsymbol{s}_0=\boldsymbol{s}$, the finite-horizon shaping return telescopes as
  \begin{equation}
  \begin{aligned}
  B_{i,T}
  &\triangleq \sum_{t=0}^{T}\gamma^t\bigl(\gamma\Phi_i(\boldsymbol{s}_{t+1})-\Phi_i(\boldsymbol{s}_{t})\bigr) \\
  &= \sum_{t=0}^{T}\gamma^{t+1}\Phi_i(\boldsymbol{s}_{t+1})-\sum_{t=0}^{T}\gamma^{t}\Phi_i(\boldsymbol{s}_{t}) \\
  &= -\Phi_i(\boldsymbol{s}_0)+\gamma^{T+1}\Phi_i(\boldsymbol{s}_{T+1}).
  \end{aligned}
  \label{eq:potential_telescoping}
  \end{equation}
  Since $\gamma\in(0,1)$ and $\Phi_i$ is bounded, the terminal term vanishes as $T\to\infty$, giving
  \begin{equation}
  \sum_{t=0}^{\infty}\gamma^t r_{\mathrm{in}}^i(\boldsymbol{o}_t^i,\boldsymbol{a}_t^i;\eta_i)=-\Phi_i(\boldsymbol{s}).
  \label{eq:potential_baseline}
  \end{equation}
  Hence, for either subset $\EuScript{S}$,
  \begin{equation}
  Q_{\mathrm{px},\EuScript{S}}^{i}(\boldsymbol{s},\boldsymbol{a}_{\EuScript{S}})
  =Q_{\mathrm{ex},\EuScript{S}}^{i}(\boldsymbol{s},\boldsymbol{a}_{\EuScript{S}})-\lambda\Phi_i(\boldsymbol{s}).
  \label{eq:potential_q_shift}
  \end{equation}
  The shift $-\lambda\Phi_i(\boldsymbol{s})$ is independent of $\boldsymbol{a}^i$ and cancels in the prefix-marginalized advantage:
  \begin{equation}
  \begin{aligned}
  A_{\mathrm{px}}^{i}
  \!=\!\bigl(\!Q_{\mathrm{ex},\EuScript{P}_i\cup\{i\}}^{i}\!-\!\lambda\Phi_i(\boldsymbol{s})\!\bigr)
   \!-\!\bigl(\!Q_{\mathrm{ex},\!\EuScript{P}_i}^{i}\!-\!\lambda\Phi_i(\boldsymbol{s})\bigr)
  \!=\!A_{\mathrm{ex}}^{i}.
  \end{aligned}
  \end{equation}
  This proves \eqref{eq:potential_no_bias}. Moreover, any action-independent residual has zero score-function contribution since $\mathbb{E}_{\boldsymbol{a}^i\sim\pi_{\theta_i}}[\nabla_{\theta_i}\log\pi_{\theta_i}(\boldsymbol{a}^i\mid\boldsymbol{o}^i)c_i]=c_i\nabla_{\theta_i}1=0$.
  \end{revblock}
  \end{IEEEproof}

\begin{theorem}[Extrinsic improvement under KL-controlled HAPPO updates]\label{thm:bias_aware}
  Let $\btheta \to \btheta'$ be one HAPPO pass with proxy advantages. \rcom{Under Assumption~\ref{ass:gfd}, this pass uses the graph-estimated advantage $\widehat A_{\mathrm{px},G_k}^{i}$ and the FD-regularized direction $g_{i,k}^{\mathrm{GFD}}$.} Assume that each accepted per-agent update satisfies the average-KL budget $\mathbb{E}_{\boldsymbol{s}\sim d_{\pi_{\boldsymbol{\theta}}}}[D_{\mathrm{KL}}(\pi_{\theta_m}(\cdot\mid\boldsymbol{o}^m)\|\pi_{\theta'_m}(\cdot\mid\boldsymbol{o}^m))]\le\delta_m$. This condition is enforced directly by the trust-region update in \eqref{eq:trust_region}; for the PPO-Clip update in \eqref{eq:happo_obj}, the theorem applies to the realized update only when the same KL budgets are satisfied through the KL monitoring/early-stopping rule described above. The extrinsic performance obeys:
  \begin{equation} \label{eq:theorem2_bound}
      J_{\text{ex}}(\pi_{\btheta'}) \geq \rcom{\widehat L^{\pi_{\btheta}}_{\mathrm{GFD}}}(\pi_{\btheta'}) - C_{\text{ex}} \sqrt{\frac{1}{2} \sum_{m=1}^M \delta_m} - \frac{M \lambda \rin}{(1 - \gamma)^2} \rcom{- \frac{M\varepsilon_{G,k}}{1-\gamma}},
  \end{equation}
  where the \rcom{graph-estimated} surrogate objective using proxy advantages is defined as:
  \begin{equation}
  \begin{aligned}
  \rcom{\widehat L^{\pi_{\btheta}}_{\mathrm{GFD}}}(\pi_{\btheta'})
    &= J_{\text{ex}}(\pi_{\btheta})  \\
     +& \frac{1}{1 - \gamma} 
      \sum_{i=1}^M 
      \E_{\substack{
          \boldsymbol{s} \sim d^{\pi_{\btheta}}\\
          \boldsymbol{a}_{\EuScript{P}_i} \sim \pi'_{\EuScript{P}_i}\\
          \boldsymbol{a}^i \sim \pi'_i}}
      \left[ \rcom{\widehat A^i_{\mathrm{px},G_k}}(\boldsymbol{s}, \boldsymbol{a}_{\EuScript{P}_i}, \boldsymbol{a}^i) \right].
  \end{aligned}
  \end{equation}
  \end{theorem}
  
  \begin{IEEEproof}
  Let $\pi = \pi_{\btheta}$ and $\pi' = \pi_{\btheta'}$. The proof only uses the attained policy $\pi'$ and the per-agent KL budgets, not the numerical optimizer used to obtain $\pi'$. Hence, it directly covers \eqref{eq:trust_region} and conditionally covers \eqref{eq:happo_obj} when the clipped step is accepted under the same KL budgets. \rcom{Thus the FD gradient alteration in \eqref{eq:gfd_direction} is accounted for through the attained policy $\pi'$ and the accepted KL budget; the argument does not assume an unregularized HAPPO direction.} By the standard TRPO performance bound, there exists a constant $C_{\text{ex}} > 0$ such that:
  \begin{equation} \label{eq:trpo_bound}
      J_{\text{ex}}(\pi') \geq L^{\text{ex}}_{\pi}(\pi') - C_{\text{ex}} \max_{\boldsymbol{s}} D_{\text{TV}}(\pi(\cdot|\boldsymbol{s}) || \pi'(\cdot|\boldsymbol{s})),
  \end{equation}
  where $L^{\text{ex}}_{\pi}(\pi') := J_{\text{ex}}(\pi) + \frac{1}{1 - \gamma} \E_{\boldsymbol{s} \sim d^{\pi}} \E_{\boldsymbol{u} \sim \pi'(\cdot|\boldsymbol{s})} [A_{\text{ex}}(\boldsymbol{s}, \boldsymbol{u})]$.
  
  Using Lemma~\ref{lem:telescoping}, the joint advantage decomposes as $A_{\text{ex}}(\boldsymbol{s}, \boldsymbol{u}) = \sum_{i=1}^M A^i_{\text{ex}}(\boldsymbol{s}, \boldsymbol{a}_{\EuScript{P}_i}, \boldsymbol{a}^i)$. \rcom{We introduce the graph-estimated proxy advantage $\widehat A^i_{\mathrm{px},G_k}$ by adding and subtracting it inside the surrogate expectation:}
  \begin{align}
      L^{\text{ex}}_{\pi}(\pi') \!\geq\! \rcom{\widehat L^{\pi_{\btheta}}_{\mathrm{GFD}}}(\pi') \!- \!\frac{1}{1 - \gamma} \sum_{i=1}^M \left| \E \left[ \rcom{\widehat A^i_{\mathrm{px},G_k}} - A^i_{\text{ex}} \right] \right|. \label{eq:bias_step}
  \end{align}
  By Lemma~\ref{lem:gap_bound} and Assumption~\ref{ass:gfd}, \rcom{the graph-estimated advantage satisfies}
  \begin{equation}
  \rcom{\big|\widehat A^i_{\mathrm{px},G_k}-A^i_{\mathrm{ex}}\big|\le \frac{\lambda \rin}{1-\gamma}+\varepsilon_{G,k}.}
  \end{equation}
  Substituting this into (\ref{eq:bias_step}):
  \begin{align}
  L^{\text{ex}}_{\pi}(\pi')
  &\geq \rcom{\widehat L^{\pi_{\btheta}}_{\mathrm{GFD}}}(\pi') - \frac{1}{1 - \gamma} \sum_{i=1}^M \left(\frac{\lambda \rin}{1 - \gamma}\rcom{+\varepsilon_{G,k}}\right) \nonumber\\
  &= \rcom{\widehat L^{\pi_{\btheta}}_{\mathrm{GFD}}}(\pi') - \frac{M \lambda \rin}{(1 - \gamma)^2}\rcom{-\frac{M\varepsilon_{G,k}}{1-\gamma}}.
  \end{align}
  For the divergence term in (\ref{eq:trpo_bound}), applying Pinsker's inequality and the joint KL decomposition from Theorem~\ref{KL_decomposition} yields:
  \begin{align}
  \vspace{-4pt}
      \E_{\bs} \left[ D_{\text{TV}}(\pi || \pi') \right] 
      \leq \sqrt{\frac{1}{2} \sum_{m=1}^M \delta_m}.
      \vspace{-4pt}
  \end{align}
  Combining these results recovers the bound in (\ref{eq:theorem2_bound}).
  \end{IEEEproof}


Recall the bilevel structure in \eqref{eq:bilevel} and the outer objective \eqref{eq:outer_obj}.
Let $\theta'(\boldsymbol{\eta})$ denote the parameters after one inner HAPPO pass with proxy advantages.


\begin{table}[t]
\centering
\caption{\revise{Model-parameter settings used in the simulations.}}
\label{tab:model-hyper}
\fontsize{9}{11}\selectfont
\setlength{\tabcolsep}{3pt}
\renewcommand{\arraystretch}{1.05}
\begin{tabularx}{\linewidth}{@{}>{\raggedright\arraybackslash}X l@{}}
\toprule
\textbf{Parameters} & \textbf{Value} \\
\midrule
Per-agent \textbf{Actor} (MAPPO/HAPPO/DG-HAPPO) & 64 and 64, Tanh \\
Team \textbf{Critic} (MAPPO/HAPPO) & 128 and 128, ReLU \\
Centralized \textbf{Actor} (CPPO) & 256 and 256, Tanh \\
Centralized \textbf{Critic} (CPPO) & 256 and 256, Tanh \\
Proxy value (per agent) & 64 and 64, ReLU \\
Intrinsic reward net (per agent) & 64 and 64, Tanh \\
Discount factor $\gamma$ & 0.99 \\
PPO epochs $K_{\text{epochs}}$ & 10 \\
Clip parameter $\epsilon_{\text{clip}}$ & 0.2 \\
Learning rates & $10^{-3}$ \\
Entropy coefficient (CPPO) & 0.02 \\
Action log-std (initial) & 0.0 \\
LIIR: $\alpha$, meta-LR, grad clip, scale & 0.4, $10^{-3}$, 2.0, 2.0 \\
\textcolor{black}{Penalty schedule: $\varrho_0$, $\eta_{\varrho}$, $\varrho_{\max}$, $\epsilon_{\rm feas}$} 
& \textcolor{black}{0.1, 1.05, 20, $10^{-3}D$} \\
\textcolor{black}{Federated distillation (FD): $B_{\mathrm{FD}}$, $\omega_{\mathrm{FD}}$, $\tau_{\mathrm{FD}}$} 
& \textcolor{black}{64, 0.5, 1.0} \\
Training episodes / update interval & 15000 / 64 \\
\bottomrule
\end{tabularx}
\end{table}

\begin{proposition}[Correctness of the LIIR meta-gradient]
\label{prop:liir_correct}
If trajectories are generated by $\pi_{\theta'(\boldsymbol{\eta})}$ (or reweighted via exact importance sampling), then
\begin{equation}
\label{eq:liir_correct}
\begin{aligned}
\nabla_{\eta_i} J_{\mathrm{ex}}(\pi_{\theta'(\boldsymbol{\eta})})
&= \mathbb{E}\!\left[
      \nabla_{\theta_i'}\log \pi_{\theta_i'}(\boldsymbol{a}^i\!\mid\!\boldsymbol{o}^i)\,
      A_{\mathrm{ex}}(\boldsymbol{s},\boldsymbol{u})
    \right] \\
&\cdot
   \mathbb{E}\!\left[
      \alpha_i\,\lambda\,
      \nabla_{\theta_i}\log \pi_{\theta_i}(\boldsymbol{a}^i\!\mid\!\boldsymbol{o}^i)\,
      \nabla_{\eta_i} r_{\mathrm{in}}^{\,i}(\boldsymbol{o}^i,\boldsymbol{a}^i;\eta_i)
    \right]
\end{aligned}
\end{equation}
matching \eqref{eq:outer_pg}--\eqref{eq:inner_sensitivity}.
\end{proposition}

\begin{IEEEproof}
Apply the chain rule $\nabla_{\eta_i} J_{\mathrm{ex}}(\pi_{\theta'(\boldsymbol{\eta})})=(\nabla_{\theta_i'} J_{\mathrm{ex}})(\nabla_{\eta_i}\theta_i')$. The first factor is the extrinsic policy gradient; the second equals the sensitivity of one PPO step w.r.t.\ $\eta_i$, which reduces to \eqref{eq:inner_sensitivity}.
\end{IEEEproof}
\begin{assumption}[\revise{Graph-FD stochastic approximation}]
\label{ass:gfd-sa}
\begin{revblock}
The dynamic graph critic and FD-regularized actor losses are Lipschitz in their trainable parameters, with stochastic-gradient noise of bounded second moment. The graph-induced advantage error satisfies
\begin{equation}
\label{eq:gfd_sa_graph_error_bound}
\bigl\|\widehat{A}_{\mathrm{ex},k}-A_{\mathrm{ex},k}\bigr\|_{L_2(d_{\pi_{\boldsymbol{\theta}_k}})}
\le \varepsilon_{G,k},
\end{equation}
and varies slowly across consecutive inner iterations:
\begin{equation}
\label{eq:gfd_sa_graph_slow_variation}
\bigl\|\widehat{A}_{\mathrm{ex},k+1}-\widehat{A}_{\mathrm{ex},k}\bigr\|_{L_2(d_{\pi_{\boldsymbol{\theta}_k}})}
\le c_G\alpha_k+\varepsilon_{G,k},
\end{equation}
for some constant $c_G>0$.
\end{revblock}
\end{assumption}
\begin{theorem}[\revise{Two-time-scale convergence with graph-FD perturbations}]
\label{thm:gfd_convergence}
\begin{revblock}
Let $\sum_k\alpha_k=\sum_k\beta_k=\infty$, $\sum_k(\alpha_k^2+\beta_k^2)<\infty$, and $\beta_k/\alpha_k\to0$. Under Assumption~\ref{ass:gfd-sa} and bounded martingale-difference gradient noise, suppose $\sum_k\alpha_k\varepsilon_{G,k}<\infty$ and either (a) $\sum_k\alpha_k\omega_{\mathrm{FD},k}<\infty$ or (b) $\omega_{\mathrm{FD},k}\equiv\omega_{\mathrm{FD}}$. Then the iterates converge to stationary points of the corresponding limiting bilevel objective:
\begin{equation}
\liminf_{k\to\infty}\mathbb{E}\big[\|\nabla_{\theta}\mathcal{L}_{\mathrm{lim}}(\theta_k;\eta_k)\|^2\!+\!\|\nabla_{\eta}J_{\mathrm{lim}}(\theta_k,\eta_k)\|^2\big]\!=\!0,
\end{equation}
where $(\mathcal{L}_{\mathrm{lim}},J_{\mathrm{lim}})$ are the original HAPPO-LIIR limits in case (a), and the FD-regularized limits in case (b), with inner objective $\mathcal{L}_{\mathrm{HAPPO}}-\omega_{\mathrm{FD}}\sum_i\mathcal{L}_{i}^{\mathrm{FD}}$.
\end{revblock}
\end{theorem}
\begin{IEEEproof}
\begin{revblock}
The graph critic error in \eqref{eq:gfd_sa_graph_error_bound} is a summable perturbation because $\sum_k\alpha_k\varepsilon_{G,k}<\infty$, while \eqref{eq:gfd_sa_graph_slow_variation} rules out non-summable critic drift. FD adds the drift $-\omega_{\mathrm{FD},k}\nabla_{\theta_i}\mathcal{L}_{i}^{\mathrm{FD}}$ in \eqref{eq:gfd_direction}: it is summable in case (a), and becomes part of the fixed FD-regularized limiting ODE in case (b). Since $\beta_k/\alpha_k\to0$, the inner recursion evolves on the fast timescale with quasi-static LIIR parameters, whereas the outer recursion tracks the resulting stationary set. Standard two-time-scale stochastic-approximation results then imply convergence to the stationary set of the relevant limiting bilevel objective.
\end{revblock}
\end{IEEEproof}

Theorem~\ref{thm:bias_aware} quantifies the cost of using proxy advantages\revise{; the bound tightens} to the extrinsic case under potential-based shaping (Proposition~\ref{prop:potential}). Proposition~\ref{prop:liir_correct} ensures that the LIIR meta-gradient aligns outer learning with the system objective.

\begin{revblock}
The analysis above provides a KL-controlled, sequential stochastic-approximation abstraction of Algorithm~\ref{alg:marl_scalable}, rather than a line-by-line proof of all PPO-Clip implementation details. Practical components such as rollout reuse, learned graph critics, FD regularization, and mini-batch updates are covered through KL monitoring, clipped ratios, bounded approximation errors, and standard stochastic-approximation conditions.
\end{revblock}

\begin{figure}[t]
\captionsetup{skip=2pt}
    \centering
    \subfloat[Perfect CSI]{
        \includegraphics[width=0.8\linewidth]{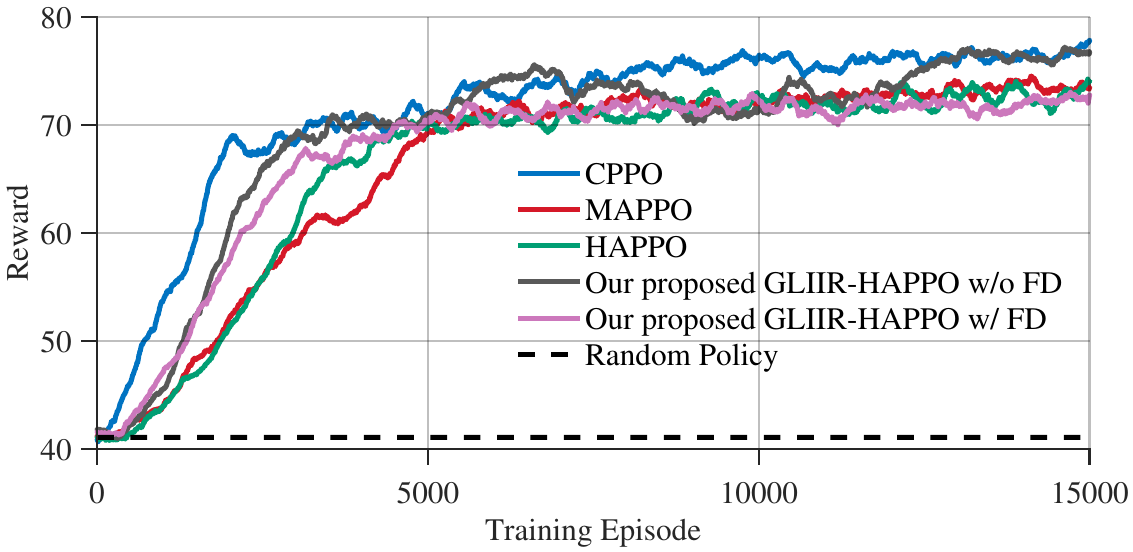}
        \label{fig:noiseA}
    }\\[-0.3ex]   
    \subfloat[Imperfect CSI, $\varepsilon=0.05$]{
        \includegraphics[width=0.8\linewidth]{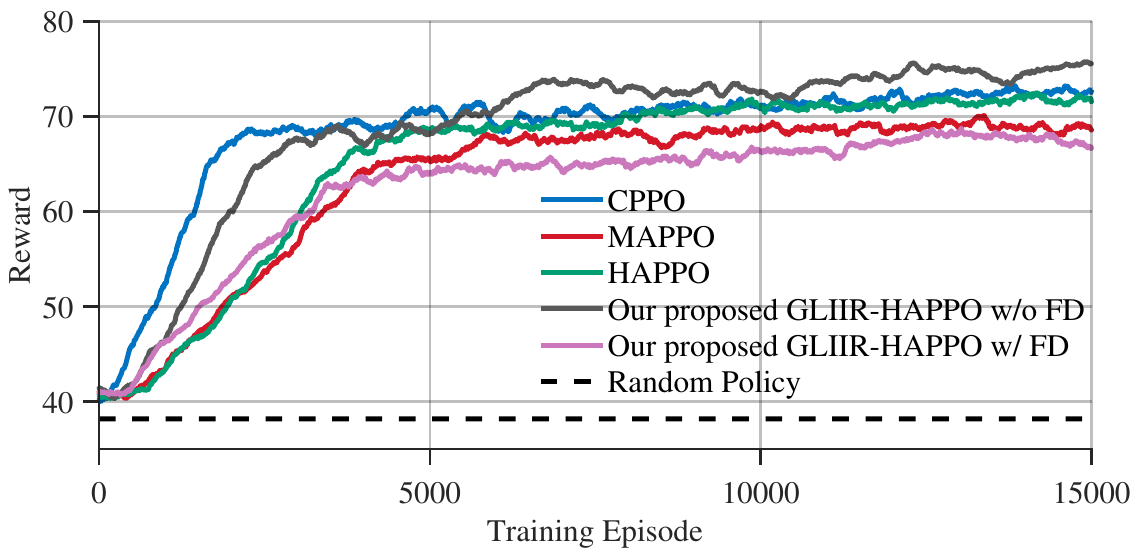}
        \label{fig:noiseB}
    }
    \caption{Convergence performance of the proposed GLIIR-HAPPO algorithm compared with benchmark schemes under (a) perfect CSI and (b) imperfect CSI with channel estimation error $\varepsilon = 0.05$.}
    \label{fig:Convergence_noise}
\end{figure}

\vspace{-0.8cm}
\section{Simulation Results}\label{sec:simulation_results}
\vspace{-0.1cm}
In this section, simulation results are presented to evaluate the performance of the proposed
GLIIR-HAPPO algorithm in MA-aided cell-free massive MIMO systems.
We consider a system with $L = 6$ APs and $K = 10$ UEs, where all entities are randomly distributed in a square area of $1000 \times 1000~\text{m}^2$ \cite{[162], 10559632}.
\begin{revblock}
Each AP is equipped with $N=4$ MAs, i.e., four movable radiating elements. This moderate-$N$ setting is used for computational tractability and controlled comparison, and it does not assume strong channel hardening.
\end{revblock}
The channel is modeled adopting a geometric approach \cite{10458417, 10243545} with $P = 10$ propagation paths at a carrier frequency of 3~GHz. The maximum downlink transmit power is $P_{l}^{\max}= 1$~W, $\forall l$, and the normalized noise power is $\sigma^2 = 0.01$. 
We utilize the regularized zero-forcing (RZF) precoding scheme \cite{4533694} with $\delta=10^{-2}$.
\begin{revblock}
  The same deterministic RZF precoding rule is adopted across all benchmark schemes. Consequently, the performance differences reported in Section V are solely attributable to the optimized power-allocation and MA-positioning strategies, ensuring a fair comparison without any confounding gains from scheme-specific precoder optimization.
\end{revblock}
To evaluate the robustness of the proposed algorithm, we consider the case of imperfect CSI, i.e., the practical channel can be modeled as $\widetilde{\boldsymbol{h}}_{lk}=\boldsymbol{h}_{lk}+\Delta \boldsymbol{h}_{lk}$ and $\Delta \boldsymbol{h}_{lk}\sim \mathcal{C N}(\mathbf{0},  \varepsilon \boldsymbol{I})$.
\revise{The model-parameter settings of the proposed GLIIR-HAPPO algorithm are shown in Table~\ref{tab:model-hyper}, and the simulations are conducted on an NVIDIA GeForce GTX 4060 GPU.}

Next, the performance of \textbf{Algorithm~\ref{alg:marl_scalable}} is compared with the following benchmark schemes:
(1) Heuristic-based schemes: particle swarm optimization (PSO) \cite{10742118}.
(2) MARL-based schemes: centralized PPO (CPPO) \cite{lu2024centralized}, MAPPO \cite{NEURIPS2022_9c1535a0}, and HAPPO \cite{Zhong2024HARL}. (3) Alternating optimization (AO)-based schemes \cite{jin2024handling}.
(4) Antenna selection (AS) \cite{1399466}: the AP adopts a larger UPA comprising $2N$ antennas ($\lambda / 2$-spaced), selecting
the optimal $N$-antenna subset via exhaustive search.

\begin{figure}[t]
\captionsetup{skip=-13pt}
	\centering
	\includegraphics[width=0.5\textwidth]{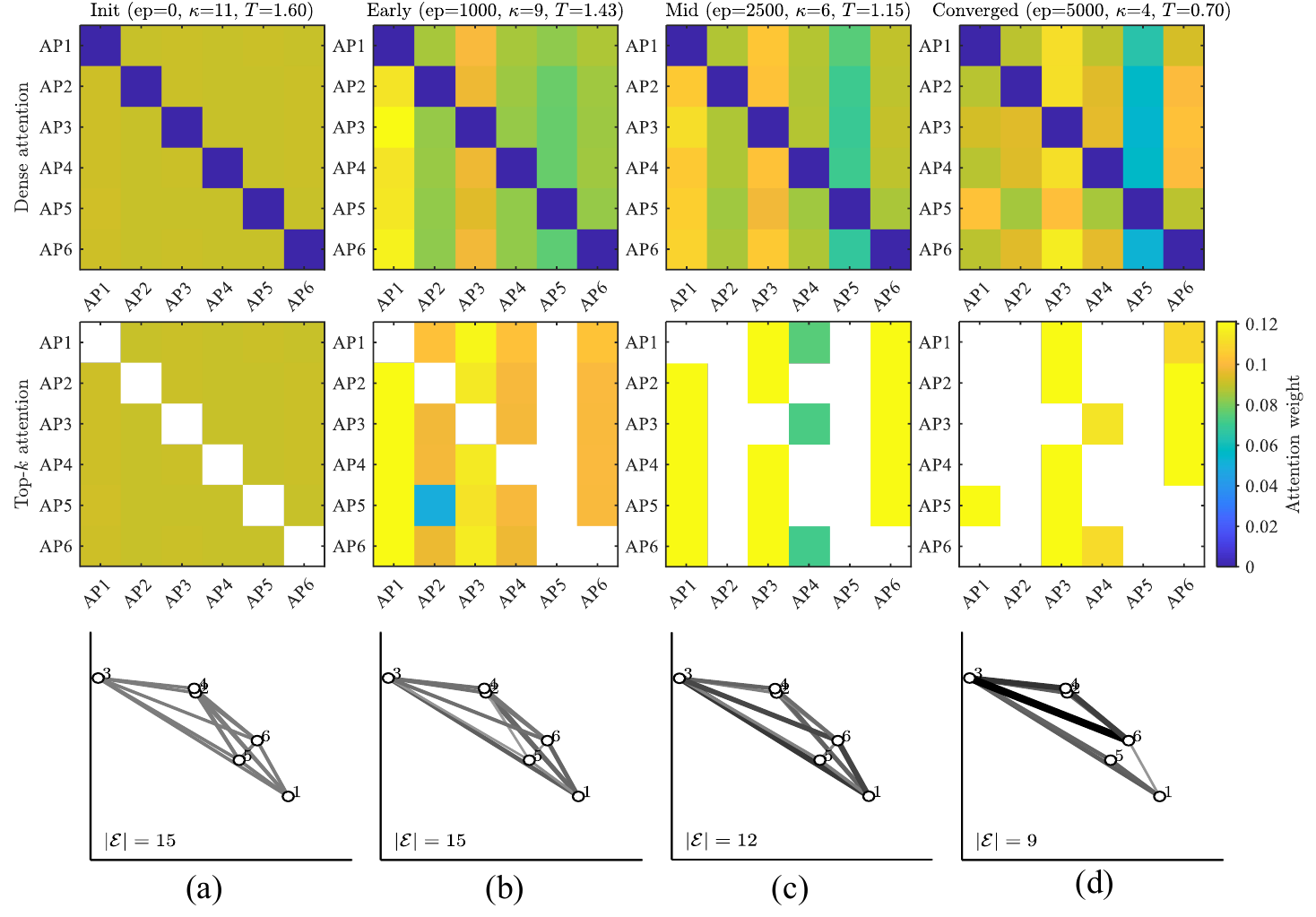}
	\caption{Visualization of the learned sparse attention patterns and AP--AP interaction graph.} 
\label{fig:heat_map}
\end{figure}

Fig.~\ref{fig:Convergence_noise} illustrates the convergence behavior of the proposed GLIIR-HAPPO algorithm under both perfect CSI and imperfect CSI with an estimation error of $\epsilon=0.05$. We benchmark our approach against CPPO, MAPPO, standard HAPPO, and a random policy.
\revise{In the perfect-CSI scenario shown in Fig.~\ref{fig:Convergence_noise}(a), the proposed framework demonstrates substantially faster convergence and a higher steady-state reward than the MAPPO and HAPPO baselines.} \scopeedit{In this simulation, the centralized CPPO baseline achieves a slightly higher final reward, while GLIIR-HAPPO achieves comparable performance with decentralized online policy execution.}

\revise{The effectiveness of the proposed architecture is further highlighted in the imperfect-CSI regime depicted in Fig.~\ref{fig:Convergence_noise}(b).} While all algorithms experience performance degradation under channel uncertainty, the proposed GLIIR-HAPPO algorithm maintains a stable and robust convergence trajectory. Notably, the inclusion of federated distillation (FD) plays a crucial role in improving training stability under noisy observations by aligning policies across homogeneous agent groups, thereby achieving the highest average reward among the distributed schemes.

\begin{table*}[t] 
  \centering
  \begingroup
  \color{black}
  \arrayrulecolor{black}
  \fontsize{8.5}{10.2}\selectfont
  \setlength{\tabcolsep}{2.2pt}
  \renewcommand{\arraystretch}{0.98}
  \makeatletter
  \def\fnum@table{\textcolor{black}{\tablename~\thetable}}
  \makeatother

  \caption{Computational Complexity of the Considered Algorithms.}
  \label{Computational_C}

  \begin{tabularx}{\textwidth}{
    @{}
    >{\raggedright\arraybackslash}p{0.29\textwidth}
    >{\centering\arraybackslash}X
    >{\centering\arraybackslash}X
    @{}
  }
    \toprule
    Algorithm 
      & Per-iteration / per-solve complexity 
      & Overall complexity \\
    \midrule

    AO (Penalty-BCD + geometric projection) 
      & \makecell[c]{Power subproblem: $\mathcal{O}\big((LK)^3\big)$;\\
        position subproblem: $\mathcal{O}(L N^3)$}
      & $\mathcal{O}\big(I_{\mathrm{AO}}\big((LK)^3 + L N^3\big)\big)$ \\[1pt]

    Heuristic (PSO) 
      & One iteration: $\mathcal{O}\big(P_{\mathrm{pop}}(D + L K N)\big)$ 
      & $\mathcal{O}\big(I_{\mathrm{PSO}} P_{\mathrm{pop}}(D + L K N)\big)$ \\[1pt]

    RL (CPPO, training)
      & \makecell[c]{One update step:\\
        $\mathcal{O}\big(P_{\mathrm{actor}}^{\mathrm{C}} 
        + P_{\mathrm{critic}}^{\mathrm{C}}\big)$}
      & \makecell[c]{$\mathcal{O}\big(E T 
        (P_{\mathrm{actor}}^{\mathrm{C}} 
        + P_{\mathrm{critic}}^{\mathrm{C}})\big)$} \\[1pt]

    RL (GLIIR-HAPPO, training) 
      & \makecell[c]{One update step:\\
        $\mathcal{O}\big(M^2 d_k 
        + M(P_{\mathrm{actor}} + P_{\mathrm{in}})\big)$}
      & \makecell[c]{$\mathcal{O}\big(E T 
        (M^2 d_k + M(P_{\mathrm{actor}} + P_{\mathrm{in}}))\big)$} \\[1pt]

    RL (GLIIR-HAPPO, online) 
      & Per time slot: $\mathcal{O}(M P_{\mathrm{actor}})$ 
      & (single-shot inference cost) \\

    \bottomrule
  \end{tabularx}

  \arrayrulecolor{black}
  \endgroup
\end{table*}

\begin{figure}[t]
\captionsetup{skip=-15pt}
	\centering
	\includegraphics[width=0.5\textwidth]{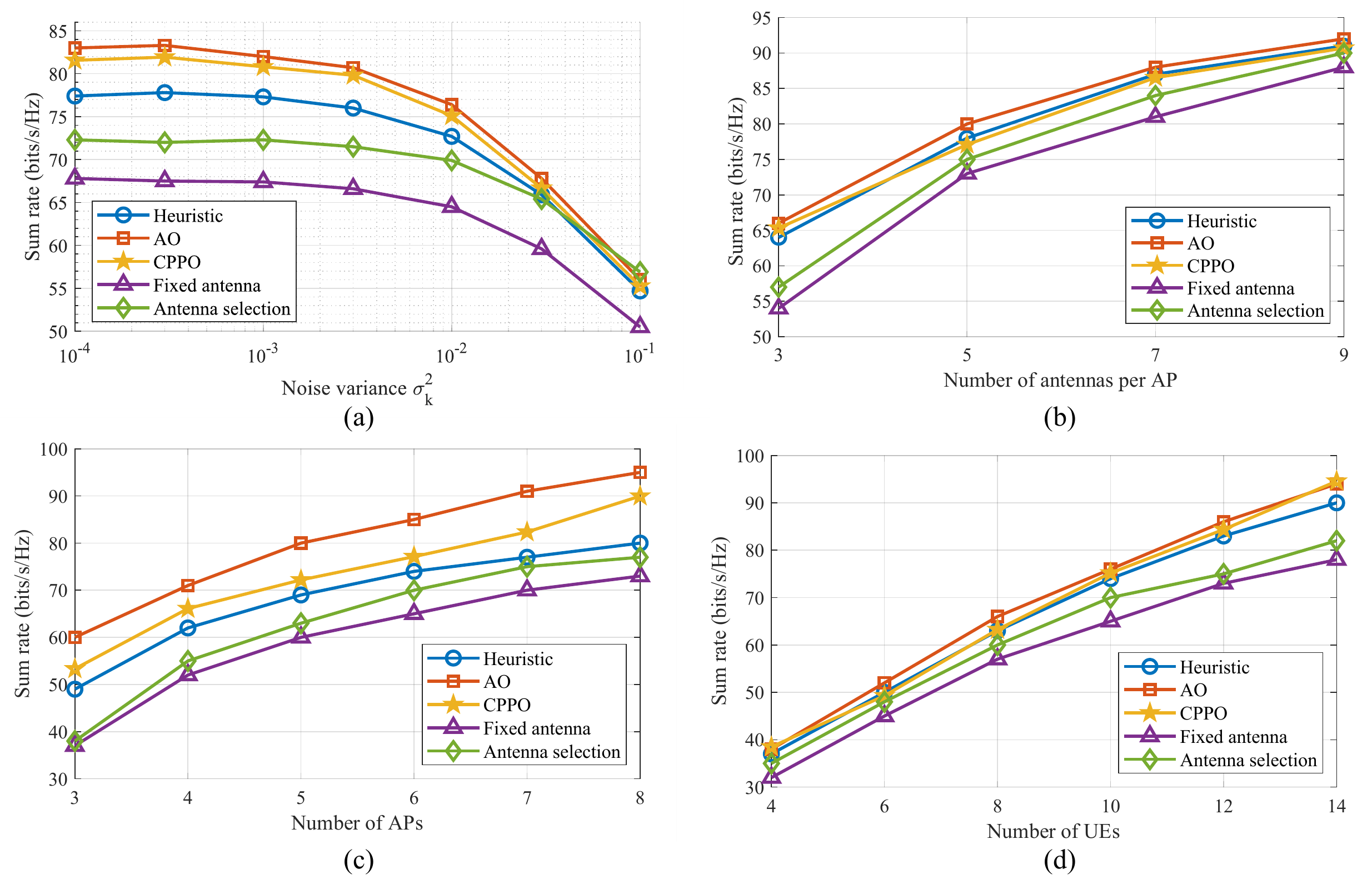}
	\caption{Performance comparison of the system sum rate against noise variance, number of antennas per AP, number of APs, and number of UEs for various benchmark algorithms.} 
\label{fig:SE_N}
\end{figure}

\begin{figure}[t]
\captionsetup{skip=-15pt}
	\centering
	\includegraphics[width=0.5\textwidth]{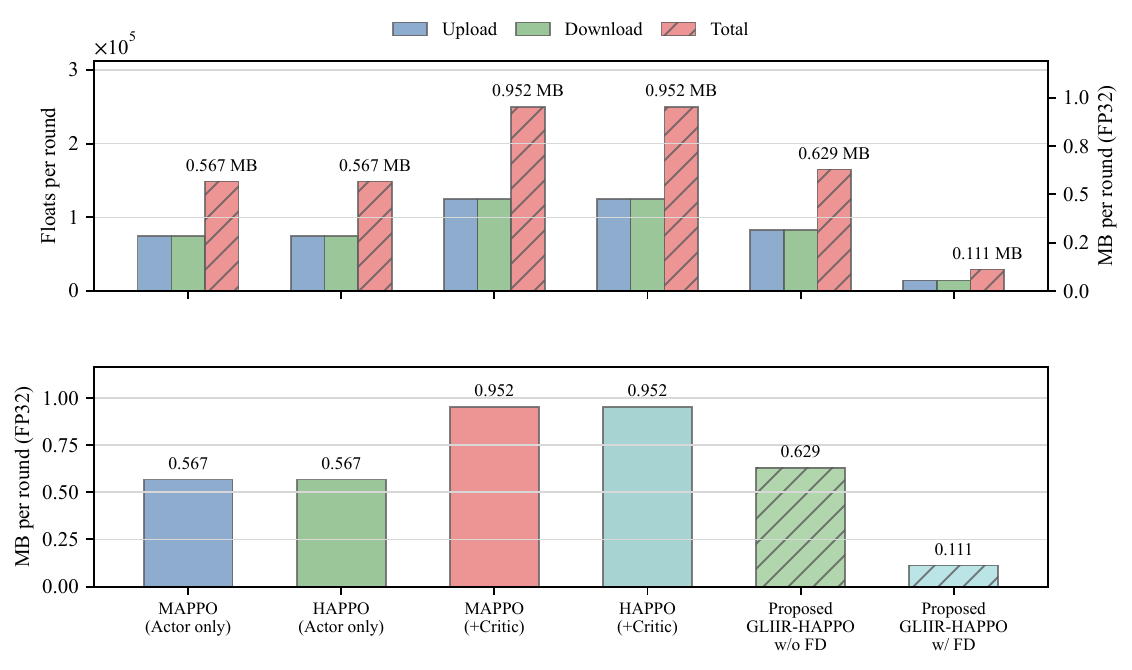}
	\caption{\scopeedit{Comparison of communication overhead per training round for different multi-agent learning schemes.}} 
\label{fig:overhead}
\end{figure}

\revise{To demonstrate the communication efficiency of the proposed architecture, Fig.~\ref{fig:heat_map} visualizes the evolution of sparse attention weights learned by the dynamic-interaction graph critic, illustrating how the attention mechanism uncovers a sparse coordination backbone during training.} In Fig.~\ref{fig:heat_map}(a)-(d), the top row shows dense attention
matrices, while the bottom row displays the corresponding top-$\kappa$
sparse graphs obtained by retaining, for each AP, only its $\kappa$
strongest outgoing neighbors. At initialization (\emph{Init}), the attention
is nearly uniform and the graph is close to fully connected, indicating that
the critic does not yet distinguish which AP-AP links are useful. As
training progresses to the \emph{Early} and \emph{Mid} stages, the
temperature $T$ and sparsity level $\kappa$ are gradually annealed, and the
attention patterns become increasingly concentrated on a few APs, while many
weak links are either heavily down-weighted or pruned. In the
\emph{Converged} stage, the interaction structure collapses to a small set
of high-weight edges forming a clear backbone that connects only a subset of
AP pairs, as reflected by the reduced edge count $|\mathcal{E}|$. This figure
highlights an important insight of our approach: the proposed dynamic
interaction critic does not merely fit value estimates, but simultaneously
performs structure learning, automatically eliminating redundant cooperation
links and uncovering a communication-efficient, interpretable graph topology
that supports the final performance.
\begin{revblock}
The sparse graphs in Fig.~4 visualize the learned attention-selected edges. In the actual training graph, these learned edges are augmented by the fixed connected backbone in \eqref{eq:connectivity_neighbor_set}. Therefore, the reduced edge count depicted in Fig. 4 represents the pruning of adaptive non-backbone links, while the retained backbone ensures that no AP becomes isolated, preserving inter-agent connectivity throughout training.
\end{revblock}

Fig.~\ref{fig:SE_N} compares the downlink sum-rate performance of the MA-aided cell-free massive MIMO network under different system configurations and optimization strategies, including the centralized CPPO baseline, heuristic PSO, AO, and the FPA and AS benchmarks. \revise{AO is a centralized per-instance solver, while CPPO is a centralized learning baseline; GLIIR-HAPPO is evaluated for distributed online execution.} As shown in Fig.~\ref{fig:SE_N}(a), all methods exhibit reduced throughput as the noise variance increases; however, CPPO maintains a clear performance margin, indicating that the learned joint power-and-position control is robust across different noise levels. Figs.~\ref{fig:SE_N}(b) and (c) show that increasing either the number of antennas per AP or the number of APs monotonically improves the sum rate for all schemes. Moreover, the gain achieved by MA is consistently larger than that of FPA and AS, confirming that the proposed framework can more effectively exploit the additional spatial degrees of freedom enabled by movable antennas than conventional designs. Finally, Fig.~\ref{fig:SE_N}(d) reveals that as the number of UEs increases and the system becomes increasingly interference-limited, the sum rate of all schemes gradually saturates; nonetheless, the MA-based methods (heuristic PSO, AO, and the centralized CPPO baseline) consistently achieve a noticeably higher sum rate than the FPA and AS benchmarks. This highlights that jointly optimizing MA positions and power allocation remains highly beneficial even under dense user deployments.


Fig.~\ref{fig:overhead} shows that adding FD to GLIIR-HAPPO reduces per-round communication overhead by more than 80\% relative to its non-FD counterpart and remains below MAPPO/HAPPO, because FD exchanges compact policy knowledge instead of raw observations or high-dimensional critic messages.
\begin{figure}[t]
\centering
\subfloat[Actuator speed]{%
\includegraphics[width=0.46\columnwidth]{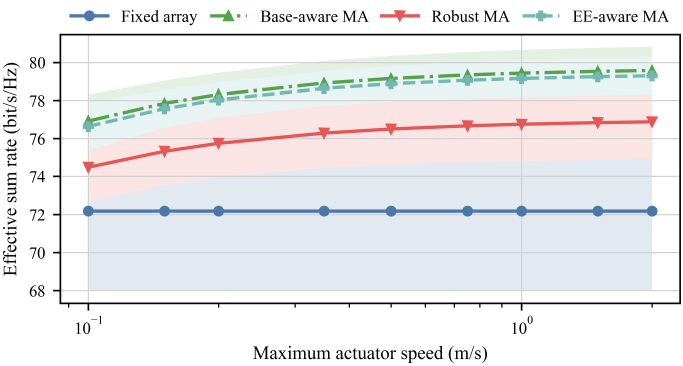}
\label{fig:practical_actuator_speed}}
\hfill
\subfloat[CSI overhead]{%
\includegraphics[width=0.46\columnwidth]{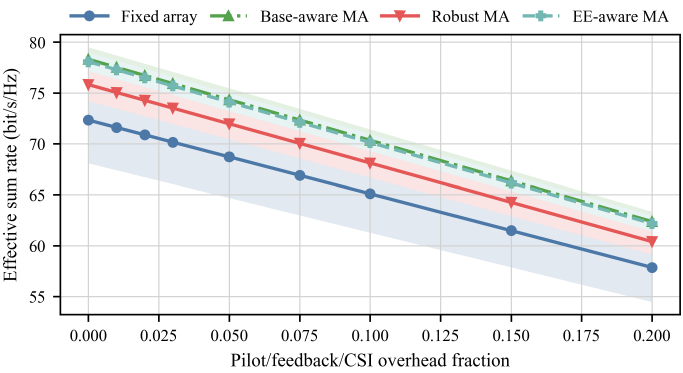}
\label{fig:practical_csi_overhead}}
\\[0.8ex]
\subfloat[Positioning error]{%
\includegraphics[width=0.46\columnwidth]{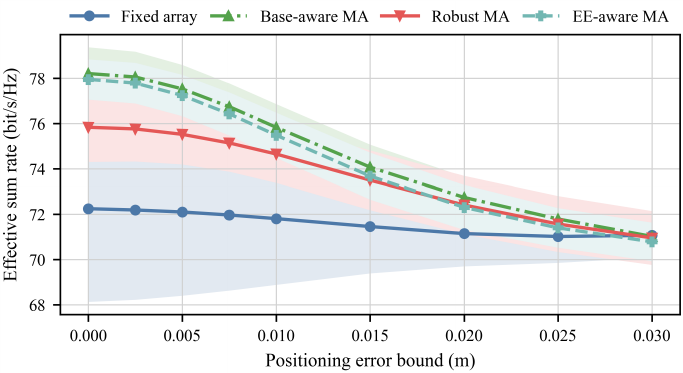}
\label{fig:practical_position_error}}
\hfill
\subfloat[EE vs CSI overhead]{%
\includegraphics[width=0.46\columnwidth]{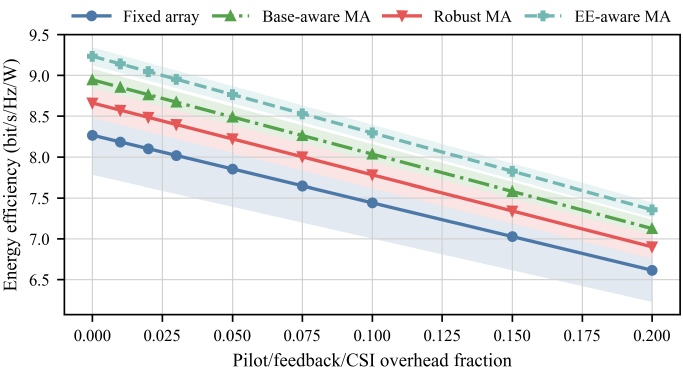}
\label{fig:practical_csi_overhead_ee}}
\caption{\textcolor{black}{Practical sensitivity versus actuator speed, CSI overhead, positioning error, and energy efficiency under CSI overhead.}}
\label{fig:practical_sensitivity_main}
\end{figure}
Fig.~\ref{fig:practical_sensitivity_main} compares four schemes: Fixed array without MA reconfiguration, Base-aware MA for nominal effective-rate optimization, Robust MA with conservative geometry under position/resolution uncertainty, and EE-aware MA for the rate-power tradeoff. Under the default setting, Base-aware MA and EE-aware MA achieve 78.20 and 77.92 bit/s/Hz, respectively, outperforming the fixed array at 72.16 bit/s/Hz. However, the effective rate decreases when actuator speed is reduced, CSI/control overhead increases, or positioning error is enlarged, as these factors reduce the effective data-transmission fraction and degrade optimized MA positioning accuracy. Fig.~\ref{fig:practical_csi_overhead_ee} further shows that energy efficiency also decreases with CSI/control overhead due to the reduced useful-transmission fraction and added control/feedback cost. Robust MA achieves a lower nominal rate but exhibits the smallest performance loss under pessimistic conditions (16.1\%), with zero distance violation, highlighting that practical MA deployment must be hardware- and robustness-aware rather than optimized solely for ideal sum-rate performance.

\scopeedit{Table~\ref{Computational_C} compares algorithmic costs. The AO power-update cost scales as $\mathcal{O}((LK)^3)$ per iteration.}
\scopeedit{For fixed actor size $P_{\text{actor}}$, aggregate GLIIR-HAPPO inference costs $\mathcal{O}(MP_{\text{actor}})$, while training retains the $\mathcal{O}(M^2d_k)$ attention cost. The reported inference complexity accounts only for policy execution and excludes the costs associated with channel acquisition, precoding, geometric feasibility processing, signaling, and physical antenna movement.}

\vspace{-1em}
\section{Conclusion} \label{conclu}
\vspace{-0.8 em}
In this paper, we proposed a novel heterogeneous MARL framework to address the joint optimization of antenna placement and transmission in MA-aided cell-free massive MIMO systems. Our approach decomposes the underlying non-convex design into role-specialized subproblems, which are handled by dedicated agents. Specifically, we handle the strict geometric constraints through a combination of variable splitting with quadratic penalties and feasibility-aware reward shaping, enabling tractable distributed learning under minimum-distance requirements. Scalable and efficient coordination is achieved via a dynamic-interaction graph critic and role-conditioned federated distillation, which together provide coordinated decision making with low communication overhead. Numerical results demonstrate that the proposed framework delivers substantial sum-rate gains over conventional FPA systems and achieves near-centralized performance while significantly reducing communication overhead. Moreover, the stability of the proposed approach is supported by trust-region analysis and two-timescale meta-gradient convergence arguments.
\scopeedit{Future work will integrate movement delay, actuation energy, positioning uncertainty, finite resolution, and CSI/control overhead into a hardware-aware objective and training reward, with model calibration and validation through over-the-air (OTA) or hardware experiments.\label{par:hardware_future_work}}

\vspace{-0.5cm}
\appendix
\vspace{-0.3cm}
\subsection{Proof of Lemma~\ref{lem:arc-removal}}
\label{proof_A}
Define $\boldsymbol{e} \triangleq (\boldsymbol{z}_{lk}-\boldsymbol{z}_{l\ell})/d$. Parametrize points on $\mathrm{Circle}_{l\ell}$ by
$\boldsymbol{z}=\boldsymbol{z}_{l\ell}+D \boldsymbol{u}$ with $\|\boldsymbol{u}\|_2=1$. Then
\begin{equation}
\|\boldsymbol{z} - \boldsymbol{z}_{lk}\|_2^2
= \| \boldsymbol{z}_{l\ell}+D\boldsymbol{u} - \boldsymbol{z}_{lk} \|_2^2
= D^2 + d^2 - 2 D d\, \boldsymbol{u}^{T} \boldsymbol{e}.
\end{equation}
Hence the points on $\mathrm{Circle}_{l\ell}$ that violate the inequality
$\|\boldsymbol{z} - \boldsymbol{z}_{lk}\|_2 \ge D$ are those for which
\begin{equation}
\vspace{-4pt}
\|\boldsymbol{z} - \boldsymbol{z}_{lk}\|_2^2 < D^2
\quad\Longleftrightarrow\quad
\boldsymbol{u}^T \boldsymbol{e} > \frac{d}{2D}.
\vspace{-2pt}
\end{equation}
For $d\ge 2D$, the inequality cannot hold and no arc is
removed. For $0<d<2D$, the set $\{\boldsymbol{u}\in\mathbb S^1 : \boldsymbol{u}^T \boldsymbol{e} > d/(2D)\}$ is a single open
arc on the unit circle, hence its image $\boldsymbol{z}=\boldsymbol{z}_{l\ell}+D\boldsymbol{u}$ is a single open arc on
$\mathrm{Circle}_{l\ell}$. For $d=0$ we have $\|\boldsymbol{z}-\boldsymbol{z}_{lk}\|_2\equiv D$ on
$\mathrm{Circle}_{l\ell}$, so nothing is removed.

The endpoints of the removed arc are exactly where equality holds, i.e.,
the solutions to the two circle $\mathrm{Circle}_{l\ell}$ and $\mathrm{Circle}_{lk}$.
Subtracting the squared equations yields
$(\boldsymbol{z} - \boldsymbol{m})^T (\boldsymbol{z}_{lk} - \boldsymbol{z}_{l\ell}) = 0$
with
$\boldsymbol{m} \triangleq \tfrac{1}{2}(\boldsymbol{z}_{l\ell}+\boldsymbol{z}_{lk})$,
so any intersection lies on the line through $m$ perpendicular to $e$. Writing
$\boldsymbol{z} = \boldsymbol{m} + \boldsymbol{s}$ with $\boldsymbol{s} \perp \boldsymbol{e}$ and imposing $\|\boldsymbol{z} - \boldsymbol{z}_{l\ell}\|_2 = D$ results
\begin{equation}
\|\boldsymbol{s}\|_2^2 + \Big(\tfrac{d}{2}\Big)^2 = D^2
\quad\Rightarrow\quad
\|\boldsymbol{s}\|_2 = \sqrt{\,D^2 - \tfrac{d^2}{4}\,}.
\end{equation}
Let $\hat {\boldsymbol{n}}_{\ell k}$ be any unit vector orthogonal to $\boldsymbol{e}$.
Then the intersection points are
\begin{equation}
\vspace{-4pt}
\boldsymbol{z} \!=\! \boldsymbol{m} \pm \sqrt{D^2 \!-\! \tfrac{d^2}{4}\,}\hat {\boldsymbol{n}}_{\ell k}
\!=\! \frac{z_{l\ell}\!+\!z_{lk}}{2} \pm \frac{1}{2}\sqrt{4D^2 - d^2} \hat{ \boldsymbol{n}}_{\ell k}.
\vspace{-4pt}
\end{equation}
When $d=2D$ the two points coalesce into the tangent point $m$; when $d>2D$ there is no
solution; and when $d=0$ we follow the manuscript's convention of treating the coincident-center
case as yielding no endpoints. Therefore $\mathcal W_{l,\ell k}$ is as stated.
\vspace{-0.4cm}
\subsection{Proof of Proposition~\ref{Permutation}} \label{proof_B}
\vspace{-0.1cm}
At time $t$, let $\boldsymbol{E}_t \in \mathbb{R}^{M \times d_e}$ be the row-stacked node embeddings $\boldsymbol{e}_i$, and define shared linear maps $\boldsymbol{W}_q \in \mathbb{R}^{d_k \times d_e}$, $\boldsymbol{W}_k \in \mathbb{R}^{d_k \times d_e}$, and $\boldsymbol{W}_v \in \mathbb{R}^{d_v \times d_e}$. Set the query, key, and value matrices as:
\begin{equation}
    \begin{cases}
    \boldsymbol{Q}_t = \boldsymbol{E}_t \boldsymbol{W}_q^T \in \mathbb{R}^{M \times d_k}, \\
    \boldsymbol{K}_t = \boldsymbol{E}_t \boldsymbol{W}_k^T \in \mathbb{R}^{M \times d_k}, \\
    \boldsymbol{V}_t = \boldsymbol{E}_t \boldsymbol{W}_v^T \in \mathbb{R}^{M \times d_v}.
    \end{cases}
    \label{eq:35_corrected}
\end{equation}
When agents are permuted by $\boldsymbol{P}$, the embeddings transform as $\boldsymbol{E}'_t = \boldsymbol{P}\boldsymbol{E}_t$. Since the weight matrices are shared and applied to the feature dimension, we have:
\begin{equation}
    (\boldsymbol{Q}'_t, \boldsymbol{K}'_t, \boldsymbol{V}'_t) = (\boldsymbol{P}\boldsymbol{Q}_t, \boldsymbol{P}\boldsymbol{K}_t, \boldsymbol{P}\boldsymbol{V}_t).
\end{equation}
The masked score matrix $\boldsymbol{C}_t = \frac{1}{\sqrt{d_k}} \boldsymbol{Q}_t \boldsymbol{K}_t^T$ transforms as:
\begin{equation}
    \begin{aligned}
    \boldsymbol{C}'_t &= \frac{1}{\sqrt{d_k}} (\boldsymbol{P}\boldsymbol{Q}_t)(\boldsymbol{P}\boldsymbol{K}_t)^T \\
    &= \boldsymbol{P} \left( \frac{1}{\sqrt{d_k}} \boldsymbol{Q}_t \boldsymbol{K}_t^T \right) \boldsymbol{P}^T = \boldsymbol{P} \boldsymbol{C}_t \boldsymbol{P}^T.
    \end{aligned}
    \label{eq:38_corrected}
\end{equation}
Both row-wise $\text{TopK}_{\kappa}$ selection and row-wise softmax are permutation equivariant operations. Specifically, for any row-permutation $\boldsymbol{P}$, $\text{softmax}_{\text{row}}(\boldsymbol{P} \boldsymbol{X} \boldsymbol{P}^T) = \boldsymbol{P} \text{softmax}_{\text{row}}(\boldsymbol{X}) \boldsymbol{P}^T$. Thus, the attention matrix transforms as $\boldsymbol{A}'_t = \boldsymbol{P} \boldsymbol{A}_t \boldsymbol{P}^T$.

The aggregated messages $\boldsymbol{M}_t = \boldsymbol{A}_t \boldsymbol{V}_t$ then transform as:
\begin{equation}
    \boldsymbol{M}'_t \!=\! \boldsymbol{A}'_t\! \boldsymbol{V}'_t \!=\! (\boldsymbol{P} \boldsymbol{A}_t \boldsymbol{P}^T)(\boldsymbol{P} \boldsymbol{V}_t) = \boldsymbol{P} \!\boldsymbol{M}_t\!,
    \label{eq:39_corrected}
\end{equation}
where we used the property $\boldsymbol{P}^T \boldsymbol{P} = \boldsymbol{I}$.
Now, consider the update layer. Let $\boldsymbol{W}_m \in \mathbb{R}^{d_u \times (d_e + d_v)}$ and $\boldsymbol{b}_m \in \mathbb{R}^{d_u}$. To ensure correct dimensionality for row-stacked inputs, the weight matrix must be applied to the right (transposed). The update equation is:
\begin{equation}
    \boldsymbol{U}_t = \sigma\left( [\boldsymbol{E}_t; \boldsymbol{M}_t] \boldsymbol{W}_m^T + \mathbf{1}\boldsymbol{b}_m^T \right) \in \mathbb{R}^{M \times d_u},
    \label{eq:37_corrected}
\end{equation}
where $\mathbf{1} \in \mathbb{R}^M$ is a column vector of ones, effectively broadcasting the bias $\boldsymbol{b}_m$ to all $M$ agents. Under permutation, the concatenation becomes $[\boldsymbol{E}'_t; \boldsymbol{M}'_t] = [\boldsymbol{P}\boldsymbol{E}_t; \boldsymbol{P}\boldsymbol{M}_t] = \boldsymbol{P}[\boldsymbol{E}_t; \boldsymbol{M}_t]$. The update transforms as follows:
\begin{align}
\vspace{-4pt}
    \boldsymbol{U}'_t &= \sigma\left( \boldsymbol{P}[\boldsymbol{E}_t; \boldsymbol{M}_t] \boldsymbol{W}_m^T + \mathbf{1}\boldsymbol{b}_m^T \right) \nonumber \\
    &= \sigma\left( \boldsymbol{P} \left( [\boldsymbol{E}_t; \boldsymbol{M}_t] \boldsymbol{W}_m^T \right) + \boldsymbol{P}\mathbf{1}\boldsymbol{b}_m^T \right)  \nonumber \\
    &= \boldsymbol{P} \sigma\left( [\boldsymbol{E}_t; \boldsymbol{M}_t] \boldsymbol{W}_m^T + \mathbf{1}\boldsymbol{b}_m^T \right) \nonumber \\
    &= \boldsymbol{P} \boldsymbol{U}_t.
    \vspace{-4pt}
\end{align}
Finally, applying the shared scalar value head $\psi$ row-wise to $\boldsymbol{U}_t$ gives the vector of per-agent state-values. Let us denote this output vector as $\boldsymbol{v}_{\text{out}} \in \mathbb{R}^M$ (where the $i$-th element corresponds to $V_i(\boldsymbol{s}_t)$). Under permutation, this vector transforms as $\boldsymbol{v}'_{\text{out}} = \boldsymbol{P} \boldsymbol{v}_{\text{out}}$. The global team value is the sum of individual values:
\vspace{-0.3cm}
\begin{equation}
V'_{\text{ex}} = \mathbf{1}^T \boldsymbol{v}'_{\text{out}} = \mathbf{1}^T (\boldsymbol{P} \boldsymbol{v}_{\text{out}}) = (\boldsymbol{P}^T \mathbf{1})^T \boldsymbol{v}_{\text{out}} = \mathbf{1}^T \boldsymbol{v}_{\text{out}} = V_{\text{ex}}.
\end{equation}
This completes the proof.
\vspace{-0.45cm}
\subsection{Proof of Lemma~\ref{lem:telescoping}}\label{proof_C}
\vspace{-0.2cm}
Define the partially marginalized functions
\begin{equation}
Q^{(i)}(\boldsymbol{s},\boldsymbol{a}_{1:i}):= \mathbb{E}_{\boldsymbol{a}_{i+1:M}\sim\pi}\!\big[\,Q_{\mathrm{ex}}(\boldsymbol{s},\boldsymbol{a}_{1:M})\,\big], 
i=0,1,\dots,M,
\end{equation}
with the conventions $Q^{(0)}(\boldsymbol{s})=\mathbb{E}_{\boldsymbol{u}\sim\pi}[Q_{\mathrm{ex}}(\boldsymbol{s},\boldsymbol{u})]=V_{\mathrm{ex}}(\boldsymbol{s})$ and $Q^{(M)}(\boldsymbol{s},\boldsymbol{a}_{1:M})=Q_{\mathrm{ex}}(\boldsymbol{s},\boldsymbol{u})$. 
Therefore
\begin{equation}
\begin{aligned}
\vspace{-4pt}
A_{\mathrm{ex}}(\boldsymbol{s},\boldsymbol{u})
&= Q^{(M)}(\boldsymbol{s},\boldsymbol{a}_{1:M}) - Q^{(0)}(\boldsymbol{s}) \\
&= \sum_{i=1}^{M} \bigl(Q^{(i)}(\boldsymbol{s},\boldsymbol{a}_{1:i})
      - Q^{(i-1)}(\boldsymbol{s},\boldsymbol{a}_{1:i-1})\bigr)
\vspace{-4pt}
\end{aligned}
\end{equation}
by a telescoping sum. It remains to identify each increment with the $i$-th prefix-marginalized advantage.

By construction, fixing the first $i$ actions at time $t=0$ and marginalizing the rest is exactly the definition of the prefix-marginalized $Q$ with subset $\EuScript{S}=\EuScript{P}_{i} \cup\{i\}$, while fixing only the first $i-1$ actions corresponds to subset $\EuScript{S}=\EuScript{P}_{i}$. Hence,
\begin{equation}
Q^{(i)}\left(\boldsymbol{s}, \boldsymbol{a}_{1: i}\right) = Q_{\mathrm{ex}, \EuScript{P}_{i} \cup\{i\}}^{i}\left(\boldsymbol{s}, \boldsymbol{a}_{\EuScript{P}_{i}}, \boldsymbol{a}_{i}\right),
\end{equation}
and
\begin{equation}
Q^{(i-1)}\left(\boldsymbol{s}, \boldsymbol{a}_{1: i-1}\right) = Q_{\mathrm{ex}, \EuScript{P}_{i}}^{i}\left(\boldsymbol{s}, \boldsymbol{a}_{\EuScript{P}_{i}}\right).
\end{equation}
Taking the difference gives
\begin{equation}
Q^{(i)}\left(\boldsymbol{s}, \boldsymbol{a}_{1: i}\right)-Q^{(i-1)}\left(\boldsymbol{s}, \boldsymbol{a}_{1: i-1}\right)=A_{\mathrm{ex}}^{i}\left(\boldsymbol{s}, \boldsymbol{a}_{\EuScript{P}_{i}}, \boldsymbol{a}_{i}\right),
\end{equation}
which proves the claim after summing over $i=1, \ldots, M$. The only properties used are linearity property of conditional expectation and the boundedness needed to interchange expectations, as codified in the regularity assumptions.
\vspace{-0.45cm}
\subsection{Proof of Theorem~\ref{KL_decomposition}}\label{proof_D}
Express \(p(\boldsymbol{u}\mid \boldsymbol{s})=\prod_{i=1}^M p_i(\boldsymbol{a}_i\mid \boldsymbol{o}_i)\) and \(q(\boldsymbol{u}\mid \boldsymbol{s})=\prod_{i=1}^M q_i(\boldsymbol{a}_i\mid \boldsymbol{o}_i)\),
where \(p_i=\pi_{\theta_i}\) and \(q_i=\pi_{\theta'_i}\).
By the definition of the KL divergence:
\begin{equation}
\begin{aligned}
D_{\mathrm{KL}}(p\|q)
&=\mathbb{E}_{\boldsymbol{u}\sim p}\!\left[\log\frac{p(\boldsymbol{u}\mid \boldsymbol{s})}{q(\boldsymbol{u}\mid \boldsymbol{s})}\right]\\
&=\mathbb{E}_{\boldsymbol{u}\sim p}\!\left[\log\frac{\prod_{i=1}^M p_i(\boldsymbol{a}_i\mid \boldsymbol{o}_i)}{\prod_{i=1}^M q_i(\boldsymbol{a}_i\mid \boldsymbol{o}_i)}\right].
\end{aligned}
\end{equation}
Using \(\log\prod_{i} x_i=\sum_{i}\log x_i\),
\begin{equation}
\vspace{-4pt}
D_{\mathrm{KL}}(p\|q)
=\sum_{i=1}^M \mathbb{E}_{\boldsymbol{u}\sim p}\!\left[\log\frac{p_i(\boldsymbol{a}_i\mid \boldsymbol{o}_i)}{q_i(\boldsymbol{a}_i\mid \boldsymbol{o}_i)}\right].
\vspace{-4pt}
\end{equation}
Under the product measure \(p(\boldsymbol{u}\mid \boldsymbol{s})=\prod_i p_i(\boldsymbol{a}_i\mid \boldsymbol{o}_i)\), the \(i\)-th summand depends only on \(\boldsymbol{a}_i\).
Hence the joint expectation reduces to the marginal expectation:
\begin{equation}
\begin{aligned}
\mathbb{E}_{\boldsymbol{u}\sim p}\!\left[\log\frac{p_i(\boldsymbol{a}_i\mid \boldsymbol{o}_i)}{q_i(\boldsymbol{a}_i\mid \boldsymbol{o}_i)}\right]
&= \mathbb{E}_{\boldsymbol{a}_i\sim p_i}\!\left[\log\frac{p_i(\boldsymbol{a}_i\mid \boldsymbol{o}_i)}{q_i(\boldsymbol{a}_i\mid \boldsymbol{o}_i)}\right] \\
&= D_{\mathrm{KL}}\!\big(p_i(\cdot\mid \boldsymbol{o}_i)\,\big\|\,q_i(\cdot\mid \boldsymbol{o}_i)\big).
\end{aligned}
\end{equation}
Summing over \(i\) establishes \eqref{eq:KL-sum}. Taking expectation in \(\boldsymbol{s}\sim d\) on both sides and using linearity of
expectation gives \eqref{eq:avg-KL}.
\vspace{-0.3cm}
\bibliographystyle{IEEEtran}
\bibliography{IEEEabrv,Ref}
\end{document}